\documentclass[11pt, letterpaper]{article}
\usepackage[margin=1in]{geometry}

\usepackage[T1]{fontenc}
\usepackage[utf8]{inputenc}
\usepackage{lmodern}
\usepackage[runin]{abstract}
\abslabeldelim{.}
\usepackage{amsmath,amssymb,amsthm,amsfonts}
\usepackage{cite}
\usepackage{hyperref}        
\usepackage{bm}              
\usepackage{stmaryrd}        
\usepackage{xspace}          
\usepackage{complexity}      
\usepackage{thmtools}        
\usepackage{macros}          
\usepackage{dashedarrow}     
\usepackage{booktabs}        
\usepackage[version=0.96]{pgf} 
\usepackage{tikz}
\usetikzlibrary{positioning}
\usetikzlibrary{arrows,automata}
\usetikzlibrary{calc}
\usetikzlibrary{petri}
\usetikzlibrary{decorations.pathreplacing}
\usetikzlibrary{decorations.pathmorphing}
\usetikzlibrary{trees}
\usetikzlibrary{shapes,shapes.geometric,fit}
\usetikzlibrary{backgrounds}

\hypersetup{
  colorlinks,
  linkcolor={black},
  citecolor={black},
  urlcolor={black}
}

\usepackage{todonotes}

\author{
  Michael Blondin\qquad Javier Esparza\qquad Stefan Jaax${}^\star$\qquad Philipp J. Meyer${}^\star$ \\[15pt]
  Technische Universität München \\
  \normalsize
  \texttt{\{blondin, esparza, jaax, meyerphi\}@in.tum.de}
}

\title{\vspace{-10pt}\textbf{Towards Efficient Verification of Population Protocols}\vspace{10pt}}
\date{}

\begin{document}

\maketitle

\begin{abstract}
  Population protocols are a well established model of computation by
  anonymous, identical finite state agents. A protocol is
  well-specified if from every initial configuration, all fair
  executions reach a common consensus. The central verification
  question for population protocols is the \emph{well-specification
    problem}: deciding if a given protocol is well-specified. Esparza
  et al. have recently shown that this problem is decidable, but with
  very high complexity: it is at least as hard as the Petri net
  reachability problem, which is \EXPSPACE-hard, and for which only
  algorithms of non-primitive recursive complexity are currently
  known.

  In this paper we introduce the class \WSSS of well-specified
  strongly-silent protocols and we prove that it is suitable for
  automatic verification. More precisely, we show that \WSSS has the
  same computational power as general well-specified protocols, and
  captures standard protocols from the literature. Moreover, we show
  that the membership problem for \WSSS reduces to solving boolean
  combinations of linear constraints over $\N$. This allowed us to
  develop the first software able to automatically prove
  well-specification for \emph{all} of the infinitely many possible
  inputs.
\end{abstract}

\section{Introduction}\label{sec:intro}

Population protocols~\cite{AADFP04,AADFP06} are a model of distributed
computation by many anonymous finite-state agents. 
They were initially introduced to model networks of passively mobile
sensors~\cite{AADFP04,AADFP06}, but are now also used to describe chemical 
reaction networks (see e.g. \cite{NB15,DBLP:conf/wdag/ChenCDS14}).

In each computation step of a population protocol, a fixed number of
agents are chosen nondeterministically, and their states are updated
according to a joint transition function. Since agents are anonymous
and identical, the global state of a protocol is completely determined
by the number of agents at each local state, called a configuration. A
protocol computes a boolean value $b$ for a given initial
configuration $C_0$ if in all fair executions starting at $C_0$, all
agents eventually agree to $b$ --- so, intuitively, population
protocols compute by reaching consensus under a certain fairness
condition. A protocol is \emph{well-specified} if it computes a value for each of
its infinitely many initial configurations (also
called \emph{inputs}).  The predicate computed by a protocol is the
function that assigns to each input the corresponding consensus
value. In a famous series of papers,
Angluin \etal~\cite{AADFP04,AADFP06} have shown that well-specified
protocols compute exactly the predicates definable in Presburger
arithmetic~\cite{AADFP04,AADFP06,AAE06,AAER07}.

In this paper we search for efficient algorithms for the
\emph{well-specification problem}: Given a population protocol, 
is it well-spe\-cified? This is a question about an infinite family
of finite-state systems. Indeed, for every input the semantics of a
protocol is a finite graph with the reachable configurations as
nodes. Deciding if the protocol reaches consensus for a fixed input
only requires to inspect one of these graphs, and can be done
automatically using a model checker. This approach has been followed
in a number of papers~\cite{PLD08,SLDP09,CMS10,guidelines}, but it only 
shows well-specification for some inputs.  There has also been work in formalizing
well-specification proofs in interactive theorem provers~\cite{DM09},
but this approach is not automatic: a human prover must first come up
with a proof for each particular protocol.

Recently, the second author, together with other co-authors, has
shown that the well-speci\-fi\-ca\-tion problem is
decidable~\cite{EGLM15}. That is, there is an algorithm that decides
if for all inputs the protocol stabilizes to a boolean value. The
proof uses deep results of the theory of Petri nets, a model very
close to population protocols. However, the same paper shows that the
well-specification problem is at least as hard as the reachability
problem for Petri nets, a famously difficult problem.  More precisely,
the problem is known to be
\EXPSPACE-hard, and all known algorithms for it have non-primitive
recursive complexity~\cite{DBLP:journals/siglog/Schmitz16}. 
In particular, there are no stable implementations of any of these algorithms, 
and they are considered impractical for nearly all applications.

For this reason, in this paper we search for a class of
well-spe\-cified protocols satisfying three properties:
\begin{itemize}
\item[(a)] \emph{No loss of expressive power}: the class should
  compute all Presburger-definable predicates.

\item[(b)] \emph{Natural}: the class should contain most protocols
  discussed in the literature.

\item[(c)] \emph{Feasible membership problem}: membership for
  the class \linebreak should have reasonable complexity.
\end{itemize}

The class \WS of all well-specified protocols obviously satisfies \linebreak
~(a) and ~(b), but not (c). So we introduce a new class \WSSS, standing for
\emph{Well-Specified Strongly Silent} protocols. We show that \WSSS
still satisfies~(a) and~(b), and then prove that the membership
problem for \WSSS is in the complexity class \DP{}; the class of
languages $L$ such that $L = L_1 \cap L_2$ for some languages
$L_1 \in \NP$ and $L_2 \in \coNP$. This is a dramatic improvement with
respect to the \EXPSPACE-hardness of the membership problem for
\WS. 

Our proof that the problem is in \DP{} reduces membership for \WSSS to
checking (un)satisfiability of two systems of boolean combinations of
linear constraints over the natural numbers. This allowed us to
implement our decision procedure on top of the constraint solver
Z3~\cite{DBLP:conf/tacas/MouraB08}, yielding the first software able
to automatically prove well-specification for \emph{all} inputs. We
tested our implementation on the families of protocols studied
in~\cite{PLD08,SLDP09,CMS10,guidelines}. These
papers prove well-specification for some inputs of protocols with up to 9
states and 28 transitions. Our approach proves well-specification for all
inputs of protocols with up to 20 states in less than one second, and
protocols with 70 states and 2500 transitions in less than one
hour. In particular, we can automatically prove well-specification for
\emph{all} inputs in less time than previous tools needed to check
\emph{one single large input}.

The verification problem for population protocols naturally 
divides into two parts: checking 
that a given protocol is well specified, and checking that a given well-specified protocol 
computes the desired predicate. While in this paper we are concerned with
well-specification, our implementation is already able to
solve the second problem for all the families of protocols described above. 
This is achieved by adding to the second system of constraints 
used to check well-specification further
linear constraints describing the sets of input configurations
for which the protocol should return {\it true} or {\it false}. An 
extension of the software that, given a protocol 
and an arbitrary Presburger predicate, checks whether the protocol computes 
the predicate, requires to solve implementation problems related to Presburger
arithmetic, and is left for future research.

The paper is organized as follows. Section \ref{sec:prelims} contains
basic definitions. Section \ref{sec:silent} introduces an intermediate
class \WSS of silent well-specified protocols, and shows that its
membership problem is still as hard as for \WS. In
Section~\ref{sec:strongly:silent}, we characterize \WSS in terms of
two properties which are then strengthened to define our new class
\WSSS. We then show that the properties defining \WSSS can be
tested in \NP{} and \coNP, and so that membership for \WSSS is in
\DP. Section~\ref{sec:strong:protocols} proves that \WSSS-protocols
compute all Presburger predicates.  Section~\ref{sec:experimental}
reports on our experimental results, and Section~\ref{sec:conclusion}
presents conclusions.

\section{Preliminaries}\label{sec:prelims}

\paragraph{Multisets.} A \emph{multiset} over a finite set $E$ is a mapping $M : E \to \N$. The set of all multisets over $E$ is denoted $\N^E$. For every $e
\in E$, $M(e)$ denotes the number of occurrences of $e$ in $M$. We
sometimes denote multisets using a set-like notation, \eg
$\multiset{f, g, g}$ is the multiset $M$ such that $M(f) = 1$, $M(g) =
2$ and $M(e) = 0$ for every $e \in E \setminus \{f, g\}$. The
\emph{support} of $M \in \N^E$ is $\supp{M} \defeq \{e \in E : M(e) >
0\}$. The \emph{size} of $M \in \N^E$ is $|M| \defeq \sum_{e \in E}
M(e)$. Addition and comparison are extended to multisets
componentwise, \ie $(M \mplus M')(e) \defeq M(e) + M'(e)$ for every $e
\in E$, and $M \leq M' \defiff M(e) \leq M(e)$ for every $e \in E$. We
define multiset difference as $(M \mminus M')(e) \defeq \max(M(e) -
M'(e), 0)$ for every $e \in E$. The empty multiset is denoted
$\vec{0}$, and for every $e \in E$ we write $\vec{e} \defeq
\multiset{e}$. 
\paragraph{Population protocols.}

A \emph{population} $P$ over a finite set $E$ is a multiset $P \in
\N^E$ such that $|P| \geq 2$. The set of all populations over $E$ is
denoted by $\pop{E}$. A \emph{population protocol} is a tuple $\PP =
(Q, T, \Sigma, I, O)$ where
\begin{itemize}
\item $Q$ is a non-empty finite set of \emph{states},
\item $T \subseteq Q^2 \times Q^2$ is a set of \emph{transitions} such
  that for every $(p, q) \in Q^2$ there exists at least a pair $(p',
  q') \in Q^2$ such that $(p, q, p', q') \in T$,
\item $\Sigma$ is a non-empty finite \emph{input alphabet},
\item $I : \Sigma \to Q$ is the \emph{input function} mapping input
  symbols to states,
\item $O : Q \to \{0, 1\}$ is the \emph{output function} mapping
  states to boolean values.
\end{itemize}

Following the convention of previous papers, we call the populations
of $\pop{Q}$ \emph{configurations}. Intuitively, a configuration $C$
describes a collection of identical finite-state \emph{agents} with
$Q$ as set of states, containing $C(q)$ agents in state $q$ for every
$q \in Q$, and at least two agents in total.

Pairs of agents\footnote{While protocols only model interactions
  between two agents, $k$-way interactions for a fixed $k > 2$ can be
  simulated by adding additional states.} interact using transitions.
For every $t = (p, q, \linebreak p', q') \in T$, we write 
$(p, q) \mapsto (p', q')$ 
to denote $t$, and we define $\pre{t} \defeq \multiset{p, q}$ and
$\post{t} \defeq \multiset{p', q'}$. For every configuration $C$ and
transition $t \in T$, we say that $t$ is \emph{enabled} at $C$ if $C
\geq \pre{t}$. Note that by definition of $T$, every configuration
enables at least one transition. A transition $t \in T$ enabled at $C$
can \emph{occur}, leading to the configuration $C \mminus \pre{t} +
\post{t}$. Intuitively, a pair of agents in states $\pre{t}$ move to
states $\post{t}$. We write $C \trans{t} C'$ to denote that $t$ is
enabled at $C$ and that its occurrence leads to $C'$. A transition $t
\in T$ is \emph{silent} if $\pre{t} = \post{t}$, 
\ie, if it cannot change the current configuration. 

For every sequence of transitions $w = t_1 t_2 \cdots t_k$, we write
$C \trans{w} C'$ if there exists a sequence of configurations
$C_0, C_1, \ldots, C_k$ such that $C = C_0 \trans{t_1} C_1 \cdots
\trans{t_k} C_k = C'$. We also write $C \trans{} C'$ if $C \trans{t}
C'$ for some transition $t \in T$, and call $C \trans{} C'$ a
\emph{step}. We write $C \trans{*} C'$ if $C \trans{w} C'$ for some $w
\in T^*$. We say that $C'$ is \emph{reachable from} $C$ if $C
\trans{*} C'$. An \emph{execution} is an infinite sequence of
configurations $C_0 C_1 \cdots$ such that $C_i \trans{} C_{i+1}$ for
every $i \in \N$. An execution $C_0 C_1 \cdots$ is \emph{fair} if for
every step $C \trans{} C'$, if $C_i = C$ for infinitely many indices
$i \in \N$, then $C_j = C$ and $C_{j+1} = C'$ for infinitely many
indices $j \in \N$. We say that a configuration $C$ is
\begin{itemize}
\item \emph{terminal} if $C \trans{*} C'$ implies $C = C'$, \ie, if
  every transition enabled at $C$ is silent;
\item a \emph{consensus configuration} if $O(p) = O(q)$ for every $p,
  q \in \supp{C}$.
\end{itemize}
For every consensus configuration $C$, let $O(C)$ denote the unique
output of the states in $\supp{C}$. An execution $C_0
C_1 \cdots$ \emph{stabilizes} to $b \in \{0, 1\}$ if there exists
$n \in \N$ such that $C_i$ is a consensus configuration and 
$O(C_i) = b$ for every $i \geq n$.

\paragraph{Predicates computable by population protocols.}

Every \emph{input} $X \in \pop{\Sigma}$ is mapped to the configuration
$I(X) \in \pop{Q}$ defined by
\begin{align*}
  I(X)(q) &\defeq \sum_{\substack{\sigma \in \Sigma \\ I(\sigma) = q}}
  X(\sigma) \text{ for every } q \in Q.
\end{align*}
A configuration $C$ is said to be \emph{initial} if $C = I(X)$ for
some input $X$. A population protocol is \emph{well-specified} if for
every input $X$, there exists $b \in \{0, 1\}$ such that every fair
execution of $\PP$ starting at $I(X)$ stabilizes to $b$. We say that
$\PP$
\emph{computes}\label{def:computes} a predicate $\varphi$ if for every
input $X$, every fair execution of $\PP$ starting at $I(X)$ stabilizes
to $\varphi(X)$. It is readily seen that $\PP$ computes a predicate if
and only if it is well-specified.

\begin{example}\label{ex:majority}
  We consider the majority protocol of \cite{AAE06} as a running
  example. Initially, agents of the protocol can be in either state
  $A$ or $B$. The protocol computes whether there are at least as many
  agents in state $B$ as there are in state $A$. The states and the
  input alphabet are $Q = \set{A, B, a, b}$ and $\Sigma = \{A, B\}$
  respectively. The input function is the identity function, and the
  output function is given by $O(B) = O(b) = 1$ and $O(A) = O(a) =
  0$. The set of transitions $T$ consists of:
  \begin{align*}
    t_{AB} & =  (A, B) \mapsto (a, b) \\ 
    t_{Ab} & =  (A, b) \mapsto (A, a) \\ 
    t_{Ba} & =  (B, a) \mapsto (B, b) \\ 
    t_{ba} & =  (b, a) \mapsto (b, b)
  \end{align*}
  and of silent transitions for the remaining pairs of states.
  Transition $t_{AB}$ ensures that every fair execution eventually reaches a
  configuration $C$ such that $C(A) = 0$ or $C(B) = 0$. If $C(A) = 0
  =C(B)$, then there were initially equally many agents in $A$ and
  $B$. Transition $t_{ba}$ then acts as tie breaker, resulting in a
  terminal configuration populated only by $b$. If, say, $C(A) > 0$
  and $C(B) = 0$, then there were initially more $A$s than $B$s, and
  $t_{Ab}$ ensures that every fair execution eventually reaches a
  terminal configuration populated only by $A$ and $a$.
\end{example}

\section{Well-specified silent protocols}\label{sec:silent}
Silent protocols\footnote{Silent protocols are also referred to as 
\emph{protocols with stabilizing states} and silent transitions are 
called \emph{ineffective} in
\cite{spirakis_stabilizing1, spirakis_stabilizing2}.} 
were introduced in \cite{dolev1999memory}. 
Loosely speaking, a protocol is silent if communication
between agents eventually ceases, \ie if every fair execution
eventually stays in the same configuration forever. Observe that a
well-specified protocol need not be silent: fair executions may keep
alternating from a configuration to another as long as they are
consensus configurations with the same output.

More formally, we say that an execution $C_0 C_1 \cdots$ is
\emph{silent} if there exists $n \in \N$ and a configuration $C$ 
such that $C_i = C$ for every
$i \geq n$. A population protocol $\PP$ is \emph{silent} if every fair
execution of $\PP$ is silent, regardless of the starting
configuration. We call a protocol that is well-specified and silent a
\emph{\WSS-protocol}, and denote by \WSS the set of all
\WSS-protocols.

\begin{example}
  As explained in Example~\ref{ex:majority}, every fair execution of
  the majority protocol is silent. This implies that the protocol is
  silent. If, for example, we add a new state $b'$ where $O(b') = 1$,
  and transitions $(b, b) \mapsto (b', b'), (b', b') \mapsto (b, b)$,
  then the protocol is no longer silent since the execution where two
  agents alternate between states $b$ and $b'$ is fair but not silent.
\end{example}

Being silent is a desirable property. While in arbitrary protocols it
is difficult to determine if an execution has already stabilized, in
silent protocols it is simple: one just checks if the current
configuration only enables silent transitions. Even though it is was
not observed explicitely, the protocols introduced in ~\cite{AADFP04} 
to characterize the expressive power of population protocols belong 
to \WSS. Therefore, \WSS-protocols can compute the same predicates 
as general ones.

Unfortunately, a slight adaptation
of~\cite[Theorem~10]{DBLP:conf/fsttcs/EsparzaGLM16} shows that the
complexity of the membership problem for \WSS-protocols is still as high
as for the general case:

\begin{restatable}{proposition}{propSilentHardness}\label{prop:silenthardness}
  The reachability problem for Petri nets is reducible in polynomial
  time to the membership problem for \WSS. In particular, membership
  for \WSS is \EXPSPACE-complete.
\end{restatable}

To circumvent this high complexity, we will show in the next section
how \WSS can be refined into a smaller class of well-speci\-fied
protocols with the same expressive power, and a membership problem of
much lower complexity.

\section{A finer class of silent well-specified protocols: \WSSS}\label{sec:strongly:silent}
It can be shown that \WSS-protocols are exactly the protocols
satisfying the two following properties:

\begin{itemize}
\item \termination: for every configuration $C$, there exists a
  terminal configuration $C'$ such that $C \trans{*} C'$.

\item \nosplitterminal: for every initial configuration $C$, there
  exists $b \in \{0, 1\}$ such that every terminal configuration $C'$
  reachable from $C$ is a consensus configuration with output $b$,
  \ie $C \trans{*} C'$ implies $O(C') = b$.
\end{itemize}

We will introduce the new class \WSSS as a refinement of \WSS obtained
by strengthening \termination and
\nosplitterminal into two new properties called \layeredtermination and
\snosplitterminal. We introduce these properties in
Section~\ref{subsec-sterm} and Section~\ref{subsec-othertwo}, and show
that their decision problems belong to \NP{} and \coNP{} respectively.

Before doing so, let us introduce some useful notions. Let $\PP = (Q,
T, \Sigma, I, O)$ be a population protocol. For every $S \subseteq T$,
$\PP[S]$ denotes the \emph{protocol induced by $S$}, \ie
$\PP[S] \defeq (Q, S \cup T', \Sigma, I, O)$ where $T' \defeq \set{(p,
q, p, q) : p, q \in Q}$ is added to ensure that any two states can
interact. Let $\trans{}_{S}$ denote the transition relation of
$\PP[S]$. An \emph{ordered partition} of $T$ is a tuple $(T_1, T_2,
\ldots, T_n)$ of nonempty subsets of $T$ such that $T =
\bigcup_{i=1}^n T_i$ and $T_i \cap T_j = \emptyset$ for every $1 \leq
i < j \leq n$.

\subsection{Layered termination}\label{subsec-sterm}

We replace \termination by a stronger property called
\layeredtermination, and show that deciding \layeredtermination
belongs to \NP. The definition of
\layeredtermination is inspired by the typical structure of protocols
found in the literature. Such protocols are organized in layers such
that transitions of higher layers cannot be enabled by executing
transitions of lower layers. In particular, if the protocol reaches a
configuration of the highest layer that does not enable any
transition, then this configuration is terminal. For such protocols,
\termination can be proven by showing that every (fair or unfair)
execution of a layer is silent.

\begin{definition}\label{def:layers}
  A population protocol $\PP = (Q, T, \Sigma, I, O)$ satisfies
  \layeredtermination if there is an ordered partition 
  $$(T_1, T_2, \ldots, T_n)$$ of $T$ such that the following properties hold for
  every $i \in [n]$:
  \begin{itemize}
  \item[(a)] For every configuration $C$, every (fair or unfair)
    execution of $\PP[T_i]$ starting at $C$ is silent.

  \item[(b)] For every configurations $C$ and $C'$, if
    $C \trans{*}_{T_i} C'$ and $C$ is terminal in $\PP[T_1 \cup
    T_2 \cup \cdots \cup T_{i-1}]$, then $C'$ is also terminal in
    $\PP[T_1 \cup T_2 \cup \cdots \cup T_{i-1}]$.

\end{itemize}
\end{definition}

\begin{example}
  The majority protocol satisfies \layeredtermination. Indeed,
  consider the ordered partition $(T_1, T_2)$, where 
  \begin{align*}
    T_1 & = \{(A, B) \mapsto (a,b), (A, b) \mapsto (A, a)\} \\
    T_2 & = \{(B, a) \mapsto (B, b), (b, a) \mapsto (b, b)\}.
  \end{align*}
  All executions of $\PP[T_1]$ and
  $\PP[T_2]$ are silent. For every terminal configuration $C$ of
  $\PP[T_1]$, we have $\supp{C} \subseteq \{A, a\}$ or $\supp{C}
  \subseteq \{B, a, b\}$. In the former case, no transition of $T_2$
  is enabled; in the latter case, taking transitons of $T_2$ cannot
  enable $T_1$.
\end{example}

As briefly sketched above, \layeredtermination implies \termination. 
In the rest of this section, we prove that checking \layeredtermination is in \NP. We do
this by showing that conditions~(a) and~(b) of
Definition~\ref{def:layers} can be tested in polynomial time.

We recall a basic notion of Petri net theory recast in the terminology
of population protocols. For every step $C \trans{t} C'$ and every
state $q$ of a population protocol, we have $C'(q) = C(q) +
\post{t}(q) - \pre{t}(q)$. This observation can be extended to
sequences of transitions. Let $|w|_t$ denote the number of occurrences
of transition $t$ in a sequence $w$. We have $C'(q) = C(q) + \sum_{t
  \in T} |w|_t \cdot (\post{t}(q) - \pre{t}(q))$. Thus, a necessary
condition for $C \trans{w} C'$ is the existence of some $\vec{x} : T
\to \N$ such that
\begin{equation}
  C'(q) = C(q) + \sum_{t \in T} \vec{x}(t) \cdot (\post{t}(q) -
  \pre{t}(q)).\label{eq:flow}
\end{equation}
We call~(\ref{eq:flow}) the \emph{flow equation} for state $q$.

\begin{proposition}\label{prop:property1}
  Let $\PP = (Q, T, \Sigma, I, O)$ be a population protocol. Deciding
  whether an ordered partition $(T_1, T_2, \ldots, T_n)$ of $T$
  satisfies condition~(a) of Definition~\ref{def:layers} can be done
  in polynomial time.
\end{proposition}

\begin{proof}
  Let $i \in [n]$ and let $U_i$ be the set of non silent transitions
  of $T_i$. It can be shown that $\PP[T_i]$ is non silent if and only
  if there exists $\vec{x} : U_i \to \Q$ such that $\sum_{t \in U_i}
  \vec{x}(t) \cdot (\post{t}(q) - \pre{t}(q)) = 0$ and $\vec{x}(q)
  \geq 0$ for every $q \in Q$, and $\vec{x}(q) > 0$ for some $q \in
  Q$. Therefore, since linear programming is in \P, we can check for
  the (non) existence of an appropriate rational solution $\vec{x}_i$
  for every $i \in [n]$.
\end{proof}

We show how to check condition~(b) of Definition~\ref{def:layers} in
polynomial time. Let $U \subseteq T$ be a set of transitions. A
configuration $C \in \pop{Q}$ is \emph{$U$-dead} if for every $t \in
U$, $C \trans{t} C'$ implies $C' = C$. We say that $\PP$ is
\emph{$U$-dead from $C_0 \in \pop{Q}$} if every configuration
reachable from $C_0$ is $U$-dead, \ie $C_0 \trans{*} C$ implies that
$C$ is $U$-dead. Finally, we say that $\PP$ is \emph{$U$-dead} if it
is $U$-dead from every $U$-dead configuration $C_0 \in \pop{Q}$.

\begin{proposition}\label{prop:property2}
  Let $\PP = (Q, T, \Sigma, I, O)$ be a population protocol. Deciding
  whether an ordered partition $(T_1, \ldots, T_n)$ of $T$ satisfies
  condition~(b) of Definition~\ref{def:layers} can be done in
  polynomial time.
\end{proposition}

\begin{proof}
  Let $i \in [n]$ and let $U \defeq T_1 \cup T_2 \cup \cdots \cup
  T_{i-1}$. $\PP[T_i]$ satisfies condition~(b) if and only if $\PP[T_i]$
  is $U$-dead. The latter can be tested in polynomial time through the
  following characterization: $\PP[T_i]$ is \emph{not} $U$-dead if and
  only if there exist $t \in T_i$ and non silent $u \in U$
  such that for every non silent $u' \in U$:
  \begin{align*}
    \pre{u'} \nleq \pre{t} + (\pre{u} \mminus \post{t}). \tag*{\qedhere}
  \end{align*}
\end{proof}

Propositions~\ref{prop:property1} and~\ref{prop:property2} yield an
\NP{} procedure to decide \layeredtermination. Indeed, it suffices to
guess an ordered partition and to check whether it satisfies
conditions~(a) and~(b) of Definition~\ref{def:layers} in polynomial
time.

\begin{corollary}\label{prop:NP}
  Deciding if a protocol satisfies \layeredtermination is in \NP.
\end{corollary}

\subsection{Strong consensus}\label{subsec-othertwo}

\newcommand{\trap}{\mathit{R}}
\newcommand{\siphon}{\mathit{S}}
\newcommand{\trapset}{\mathbf{T}}
\newcommand{\siphonset}{\mathbf{S}}

To overcome the high complexity of reachability in population
protocols, we strengthen \nosplitterminal by replacing the
reachability relation in its definition by
an \emph{over-approximation}, \ie, a relation $
\ptrans{}$ over configurations such that $C \trans{*} C'$ implies $C
\ptrans{} C'$. Observe that the flow equations provide an
over-approxi\-ma\-tion of the reachability relation. Indeed, as
mentioned earlier, if $C \trans{*} C'$, then there exists $\vec{x} : T
\to \N$ such that $(C, C', \vec{x})$ satisfies all of the flow
equations. However, this over-approximation alone is too crude for the
verification of protocols.

\begin{example}\label{ex:maj:flow}
  For example, let us consider the configurations $C = \multiset{A,
    B}$ and $C' = \multiset{a, a}$ of the majority protocol. The flow
  equations are satisfied by the mapping $\vec{x}$ such that
  $\vec{x}(t_{AB}) = \vec{x}(t_{Ab}) = 1$ and $\vec{x}(t_{Ba}) =
  \vec{x}(t_{ba}) = 0$. Yet, $C \trans{*} C'$ does not hold.
\end{example}

To obtain a finer reachability over-approximation, we introduce
so-called traps and siphons constraints borrowed from the theory of
Petri
nets~\cite{Desel:1995:FCP:207572,DBLP:journals/fmsd/EsparzaM00,DBLP:conf/cav/EsparzaLMMN14}
and successfully applied to a number of analysis problems (see
\eg~\cite{DBLP:journals/fmsd/EsparzaM00,DBLP:conf/cav/EsparzaLMMN14,DBLP:conf/cade/AthanasiouLW16}). Intuitively,
for some subset of transitions $U \subseteq T$, a $U$-trap is a set of
states $P \subseteq Q$ such that every transition of $U$ that removes
an agent from $P$ also moves an agent into $P$. Conversely, a
$U$-siphon is a set $P \subseteq Q$ such that every transition of $U$
that moves an agent into $P$ also removes an agent from $P$. More
formally, let $\preset{R} \defeq \{t \in T : \supp{\post{t}} \cap R
\not= \emptyset\}$ and $\postset{R} \defeq \{t \in T : \supp{\pre{t}}
\cap R \not= \emptyset\}$. $U$-siphons and $U$-traps are defined as
follows:

\begin{definition}\label{def:trap:siphon}
  A subset of states $P \subseteq Q$ is a \emph{$U$-trap} if
  $\postset{P} \cap U \subseteq \preset{P}$, and a \emph{$U$-siphon}
  if $\preset{P} \cap U \subseteq \postset{P}$.
\end{definition}

For every configuration $C \in \pop{Q}$ and $P \subseteq Q$, let $C(P)
\defeq \sum_{q \in P} C(q)$. Consider a sequence of steps $C_0
\trans{t_1} C_1 \trans{t_2} \cdots \trans{t_n} C_n$ where $t_1,
\ldots, t_n \in U$. It follows from Definition~\ref{def:trap:siphon} that if some transition $t_i$ moves an agent to a
$U$-trap $P$, then $C_j(P) > 0$ for every $j \geq i$. Similarly, if
some transition $t_i$ removes an agent from a $U$-siphon, then $C_j(P)
> 0$ for every $j < i$. In particular:

\begin{observation}\label{lem:basicst}
  Let $U \subseteq T$, and let $C$ and $C'$ be configurations such
  that $C \trans{*}_U C'$. For every $U$-trap $P$, if $C'(P) = 0$,
  then $\preset{P} \cap U = \emptyset$. For every $U$-siphon $P$, if
  $C(\siphon) = 0$, then $\postset{P} \cap U = \emptyset$.
\end{observation}

We obtain a necessary condition for $C \trans{*}_U C'$ to hold, which
we call potential reachability:

\begin{definition}\label{def:potreach}
  Let $C, C'$ be two configurations, let $\vec{x} : T \to \N$, and let
  $U = \supp{\vec{x}}$. We say that $C'$ is \emph{potentially
    reachable from $C$ through $\vec{x}$}, denoted $C \ptrans{\vec{x}}
  C'$, if
  \begin{itemize}
  \item[(a)] the flow equation~(\ref{eq:flow}) holds for every $q \in Q$,
  \item[(b)] $C'(P) = 0$ implies $\preset{P} \cap U = \emptyset$ for
    every $U$-trap $P$, and
  \item[(c)] $C(P) = 0$ implies $\postset{P} \cap U = \emptyset$ for
    every $U$-siphon $P$.
  \end{itemize}
\end{definition}

\begin{example}
  Let us reconsider Example~\ref{ex:maj:flow}. Let $U = \supp{\vec{x}}
  = \{t_{AB}, t_{Ab}\}$ and $P = \{A, b\}$. Recall that $t_{AB} = (A,
  B) \mapsto (a, b)$ and $t_{Ab} = (A, b) \mapsto (A, a)$. We have
  $\postset{P} \cap U = U$ which implies that $P$ is a $U$-trap. This
  means that Definition~\ref{def:potreach}(b) is violated as $C'(P) =
  0$ and $\preset{P} \cap U = U \not= \emptyset$. Therefore,
  $\multiset{A, B} \ptrans{\vec{x}} \multiset{a, a}$ does not hold.
\end{example}

We write $C \ptrans{} C'$ if $C \ptrans{\vec{x}} C'$ for some $\vec{x}
: T \to \N$. As an immediate consequence of
Observation~\ref{lem:basicst}, for every configurations $C$ and $C'$,
if $C \trans{*} C'$, then $C \ptrans{} C'$. This allows us to
strengthen \nosplitterminal by redefining it in terms of potential
reachability instead of reachability:

\begin{definition}
  A protocol satisfies \snosplitterminal if for every initial
  configuration $C$, there exists $b \in \{0, 1\}$ such that every
  terminal configuration $C'$ potentially reachable from $C$ is a
  consensus configuration with output $b$, \ie $C \ptrans{} C'$
  implies $O(C') = b$.
\end{definition} 

Since the number of $U$-traps and $U$-siphons of a protocol can be
exponential in the number of states, checking trap and siphon
constraints by enumerating them may take exponential
time. Fortunately, this can be avoided. By definition, it follows that
the union of two $U$-traps is again a $U$-trap, and similarly for
siphons. Therefore, given a configuration $C$, there exists a unique
maximal $U$-siphon $P_\text{max}$ such that $C(P_\text{max}) = 0$, and
a unique maximal $U$-trap $P'_\text{max}$ such that $C(P'_\text{max})
= 0$. Moreover, $P_\text{max}$ and $P'_\text{max}$ can be computed in
linear time by means of a simple greedy algorithm (see
\eg~\cite[Ex.~4.5]{Desel:1995:FCP:207572}). This simplifies the task
of checking traps and siphons constraints, and yields a \coNP{}
procedure for testing \snosplitterminal:

\begin{restatable}{proposition}{strconscoNP}\label{strconscoNP}
  Deciding if a protocol satisfies \snosplitterminal is in \coNP.
\end{restatable}

\begin{proof}
  Testing whether a protocol \emph{does not} satisfy \snosplitterminal
  can be done by guessing $C_0, C, C' \in \pop{Q}$, $b \in \{0, 1\}$,
  $q, q' \in Q$ and $\vec{x}, \vec{x}' : T \to \N$, and testing
  whether
  \begin{itemize}
  \item[(a)] $C_0$ is initial, $C$ is terminal, $C'$ is terminal, $q
    \in \supp{C}$, $q' \in \supp{C'}$, $O(q) \not= O(q')$, and
  \item[(b)] $C_0 \ptrans{\vec{x}} C$ and $C_0 \ptrans{\vec{x}'} C'$.
  \end{itemize}
  
  Since there is no \emph{a priori} bound on the size of $C_0, C, C'$
  and $\vec{x}, \vec{x}'$, we guess them carefully. First, we guess
  whether $D(p) = 0$, $D(p) = 1$ or $D(p) \geq 2$ for every $D \in
  \{C_0, C, C'\}$ and $p \in Q$. This gives enough information to
  test~(a). Then, we guess $\supp{\vec{x}}$ and
  $\supp{\vec{x}'}$. This allows to test traps/siphons constraints as
  follows. Let $U \defeq \supp{\vec{x}}$, let $P_\text{max}$ be the
  maximal $U$-trap such that $C(P_\text{max}) = 0$, and let
  $P'_\text{max}$ be the maximal $U$-siphon such that
  $C_0(P'_\text{max}) = 0$. Conditions~(b) and (c) of
  Definition~\ref{def:potreach} hold if and only if
  $\preset{(P_\text{max})} \cap U = \emptyset$ and
  $\postset{(P'_\text{max})} \cap U = \emptyset$, which can be tested
  in polynomial time. The same is done for $\vec{x}'$. If~(a) and
  siphons/traps constraints hold, we build the system $\mathcal{S}$ of
  linear equations/inequalities obtained from the conjunction of the
  flow equations together with the constraints already guessed. By
  standard results on integer linear programming (see
  \eg~\cite[Sect.~17]{Sch86}), if $\mathcal{S}$ has a solution, then
  it has one of polynomial size, and hence we may guess it.
\end{proof}

\subsection{\WSSS-protocols}\label{subsec-strongconvergence}

We say that a protocol belongs to \WSSS if it satisfies
\layeredtermination and \snosplitterminal. Since
\WSSS$\subseteq$\WSS$\subseteq$\WS holds, every \WSSS-protocol is
well-specified. Recall that a language $L$ belongs to the class
\DP{}~\cite{DBLP:books/daglib/0018514} if there exist languages $L_1
\in \NP$ and $L_2 \in \coNP$ such that $L = L_1 \cap L_2$. By taking
$L_1$ and $L_2$ respectively as the languages of population protocols
satisfying \layeredtermination and \snosplitterminal,
Corollary~\ref{prop:NP} and Proposition~\ref{strconscoNP} yield:

\begin{theorem}\leavevmode 
  The membership problem for \WSSS-protocols is in \DP.
\end{theorem}

\section{\WSSS is as expressive as \WS}\label{sec:strong:protocols}
In a famous result, Angluin~\etal~\cite{AAE06} have shown that a
predicate is computable by a population protocol if and only if it is
definable in Presburger arithmetic, the first-order theory of
addition \cite{AADFP04,AAE06}. In
particular, \cite{AADFP04}~constructs protocols for
Presburger-definable predicates by means of a well-known result:
\emph{Presburger-definable} predicates are the smallest set of
predicates containing all threshold and remainder predicates, and
closed under boolean operations. A
\emph{threshold predicate} is a predicate of the form 
$$P(x_1, \ldots,x_k) = \left(\sum_{i=1}^k a_i x_i < c\right),$$ 
where $k \geq 1$ and $a_1, \ldots, a_k, c \in \Z$. 
A \emph{remainder predicate} is a predicate of the form 
$$P(x_1, \ldots, x_k) = \left(\sum_{i=1}^k a_i x_i \equiv c\ (\mathrm{mod}\ m)\right),$$ 
where $k \geq 1$, $m \geq 2$ and $a_1, \ldots, a_k, c, m \in \Z$. 
Here, we show that these predicates can be
computed by \WSSS-protocols, and that \WSSS is closed under negation
and conjunction. As a consequence, we obtain that \WSSS-protocols are
as expressive as \WS, the class of all well-specified protocols.

\newcommand{\lmax}{v_\text{max}}
\newcommand{\val}[1]{\mathrm{val}(#1)}

\paragraph{Threshold.} We describe the protocol given
in~\cite{AADFP04} to compute the threshold predicate $\sum_{i=1}^k a_i
x_i < c$.  Let
$$\lmax \defeq \max(|a_1|, |a_2|, \ldots, |a_k|, |c|+1)$$ and define
\begin{align*}
  f(m, n) & \defeq \max(-\lmax, \min(\lmax, m + n)) \\
  g(m, n) & \defeq (m + n) - f(m, n) \\
  b(m, n) & \defeq (f(m, n) < c)
\end{align*}
The protocol is $\PP_{\text{thr}} \defeq (Q, T, \Sigma, I, O)$, where
\begin{align*}
  Q & \defeq \{0, 1\} \times [-\lmax, \lmax] \times \{0, 1\} \\
  \Sigma & \defeq \{x_1, x_2, \ldots, x_k\} \\
  I(x_i) & \defeq (1, a_i, a_i < c) \text{ for every } i \in [k] \\
  O(\ell, n, o) & \defeq o \text{ for every state } (\ell, n, o),
\end{align*}
and $T$ contains
\begin{align*}
  (1, n, o), (l, n', o') & \mapsto (1, f(n, n'), b(n, n')), (0,
  g(n, n'), b(n, n'))
\end{align*}
for every $n, n' \in [-\lmax, \lmax]$, $\ell, o, o' \in \{0,1\}$.
Intuitively, a state $(\ell, n, o)$ indicates that the agent has value
$n$, opinion $o$, and that it is a leader if and only if $\ell =
1$. When a leader $q$ and a state $r$ interact, $r$ becomes a non
leader, and $q$ increases its value as much as possible by
substracting from the value of $r$. Moreover, a leader can change the
opinion of any non leader.

\begin{restatable}{proposition}{propThresholdConsensus}\label{prop:threshold:strong:consensus}
  $\PP_{\text{thr}}$ satisfies \snosplitterminal.
\end{restatable}

\begin{proof}
  Let $\val{q} \defeq n$ for every state $q = (\ell, n, o) \in Q$, and
  let $\val{C} \defeq \sum_{q \in Q} C(q) \cdot \val{q}$ for every
  configuration $C \in \pop{Q}$. The following holds for
  every $C, C' \in \pop{Q}$:
  \begin{itemize}
    \item[(a)] If $(C, C', \vec{x})$ is a solution to the flow
      equations for some $\vec{x} : T \to \N$, then $\val{C} =
      \val{C'}$.

    \item[(b)] If $C, C'$ are terminal, $C$ and $C'$ contain a leader,
      and $\val{C} = \val{C'}$, then $O(C) = O(C')$.
  \end{itemize}
  Suppose for the sake of contradiction that $\PP$ does not satisfy \linebreak
  \snosplitterminal. There exist $C_0, C, C' \in \pop{Q}$, $q, q' \in
  Q$ and $\vec{x}, \vec{x}' : T \to \N$ such that $C_0
  \ptrans{\vec{x}} C$, $C_0 \ptrans{\vec{x}'} C'$, $C_0$ is initial,
  $C$ and $C'$ are terminal consensus configurations, $q \in
  \supp{C}$, $q' \in \supp{C'}$ and $O(q) \neq O(q')$. Note that
  $(C_0, C, \vec{x})$ and $(C_0, C', \vec{x}')$ both satisfy the flow
  equations. Thus, by~(a), $\val{C} = \val{C_0} = \val{C'}$. Since
  $C_0$ is initial, it contains a leader. Since the set of leaders
  forms a $U$-trap for every $U \subseteq T$, and $(C_0, C, \vec{x})$
  and $(C_0, C', \vec{x})$ satisfy trap constraints, $C$ and $C'$
  contain a leader. By~(b), $C$ and $C'$ are consensus configurations
  with $O(C) = O(C')$, which is a contradiction.
\end{proof}

\begin{restatable}{proposition}{propThresholdLayered}\label{prop:threshold:layered}
  $\PP_{\text{thr}}$ satisfies \layeredtermination.
\end{restatable}

\begin{proof}
  Let $L_0 \defeq \{(1, x, 0) : c \leq x \leq \lmax\}$, $L_1 \defeq
  \{(1, x, 1) : -\lmax \leq x < c\}$, $N_0 \defeq \{(0, 0, 0)\}$ and
  $N_1 \defeq \{(0, 0, 1)\}$. It can be shown that the following
  ordered partitions satisfy layered termination for $c > 0$ and $c
  \leq 0$ respectively:

  \begin{align*}
     T_1 & \defeq \{t \in T : \pre{t} \neq \multiset{q, r} \text{ for all } q
     \in L_0, r \in N_1\}, \\
     T_2  & \defeq T \setminus T_1, \text{ and} \\[5pt]
     S_1 & \defeq \{t \in T : \pre{t} \neq \multiset{q, r} \text{ for all } q
     \in L_1, r \in N_0\}, \\
     S_2 &\defeq T \setminus S_1. \qedhere
  \end{align*}
\end{proof}

\newcommand{\true}{\mathrm{true}}
\newcommand{\false}{\mathrm{false}}

\paragraph{Remainder.} We give a protocol for the remainder predicate
$$\sum_{i=1}^k a_i x_i \equiv c\; (\text{mod } m).$$ The protocol is
$\PP_{\text{rmd}} = (Q, T, \Sigma, I, O)$, where
\begin{align*}
  Q & \defeq [0, m) \cup \{\true, \false\} \\
  \Sigma & \defeq \{x_1, x_2, \ldots, x_k\} \\
  I(x_i) & \defeq a_i \text{ mod } m \text{ for every } i \in [k] \\
  O(q) & \defeq \begin{cases}
                  1  \text{ if } q \in \{c, \true\} \\
                  0  \text{ otherwise }
                \end{cases}
\end{align*}
  and where $T$ contains
  the following transitions for every $n, n' \in [0, m)$ and $b \in
    \{\false, \true\}$:
  \begin{align*}
  (n, n') & \mapsto (n + n' \text{ mod } m, n + n' \text{ mod } m = c) \quad \mbox{ and
      } \\
  (n, b) & \mapsto (n, n = c).
  \end{align*}
  In the appendix we show
  that $\PP_{\text{rmd}}$ belongs to \WSSS by adapting the proof
  for $\PP_{\text{thr}}$.
\label{ssec:basic:protocols}

\paragraph{Negation and conjunction.} Let $\PP_1 = (Q_1, T_1, \Sigma, I_1, O_1)$ and $\PP_2 = (Q_2, T_2, \Sigma, I_2, O_2)$ be \WSSS-protocols computing predicates $\varphi_1$
and $\varphi_2$ respectively.
We may assume that $\PP_1$ and $\PP_2$ are defined over identical $\Sigma$,
for we can always extend the input domain of threshold/remainder predicates
by variables with coefficients of value zero.
The predicate $\neg \varphi_i$ can be
computed by replacing $O_i$ by the new output function $O_i'$ such
that $O_i'(q) \defeq \neg O_i(q)$ for every $q \in Q_i$. To compute
$\varphi_1 \land \varphi_2$, we build an asynchronous product where
steps of $\PP_1$ and $\PP_2$ can be executed independently.

More formally, the \emph{conjunction} of $\PP_1$ and $\PP_2$ is defined as
the population protocol $\PP \defeq (Q, S, I, \Sigma, O)$ such that $Q
\defeq Q_1 \times Q_2$, $S \defeq S_1 \cup S_2$, $I(\sigma) \defeq
(I_1(\sigma), I_2(\sigma))$ and $O(p, q) \defeq O_1(p) \land O_2(q)$
where
\begin{align*}
  S_1 &\defeq \{(p, r), (p', r') \mapsto (q, r), (q', r') :  (p, p', q, q') \in
 T_1, r, r' \in Q_2\}, \\
  S_2 & \defeq \{(r, p), (r', p') \mapsto (r, q), (r', q') : (p, p', q, q') \in
 T_2, r, r' \in Q_1\}.
\end{align*}
In the appendix
we show that $\PP$ is in \WSSS since terminal/consensus configurations,
flow equations, and traps and siphons
constraints are preserved by projections from $\PP$ onto $\PP_1$ and
$\PP_2$.
\label{ssec:boolean:protocols}

\begin{table}[h]
\centering
\makebox[0pt][c]{\scalebox{0.8}{\parbox{1.25\textwidth}{%
    \begin{minipage}[t]{0.24\hsize}\centering
        \begin{tabular}[t]{rrrr}
            \toprule
                \multicolumn{4}{c}{Threshold} \\
            \midrule
                \multicolumn{1}{c}{$\lmax$} & \multicolumn{1}{c}{$|Q|$} & \multicolumn{1}{c}{$|T|$} & \multicolumn{1}{c}{Time} \\
            \midrule
                 3  & 28 &  288 &     8.0 \\
                 4  & 36 &  478 &    26.5 \\
                 5  & 44 &  716 &    97.6 \\
                 6  & 52 & 1002 &   243.4 \\
                 7  & 60 & 1336 &   565.0 \\
                 8  & 68 & 1718 &  1019.7 \\
                 9  & 76 & 2148 &  2375.9 \\
                10  & 84 & 2626 & timeout \\
            \bottomrule
        \end{tabular}
    \end{minipage}%
    \hfill
    \begin{minipage}[t]{0.24\hsize}\centering
        \begin{tabular}[t]{rrrr}
            \toprule
                \multicolumn{4}{c}{Remainder} \\
            \midrule
                \multicolumn{1}{c}{$m$} & \multicolumn{1}{c}{$|Q|$} & \multicolumn{1}{c}{$|T|$} & \multicolumn{1}{c}{Time} \\
            \midrule
                10  & 12 &   65 &     0.4 \\
                20  & 22 &  230 &     2.8 \\
                30  & 32 &  495 &    15.9 \\
                40  & 42 &  860 &    79.3 \\
                50  & 52 & 1325 &   440.3 \\
                60  & 62 & 1890 &  3055.4 \\
                70  & 72 & 2555 &  3176.5 \\
                80  & 82 & 3320 & timeout \\
            \bottomrule
            \addlinespace[0.5em]
            \toprule
                \multicolumn{4}{c}{Majority} \\
            \midrule
                & \multicolumn{1}{c}{$|Q|$} & \multicolumn{1}{c}{$|T|$} & \multicolumn{1}{c}{Time} \\
            \midrule
                & 4 & 4 & 0.1 \\
            \bottomrule
        \end{tabular}
    \end{minipage}%
    \hfill
    \begin{minipage}[t]{0.24\hsize}\centering
        \begin{tabular}[t]{rrrr}
            \toprule
                \multicolumn{4}{c}{Flock of birds~\cite{CMS10}} \\
            \midrule
                \multicolumn{1}{c}{$c$} & \multicolumn{1}{c}{$|Q|$} & \multicolumn{1}{c}{$|T|$} & \multicolumn{1}{c}{Time} \\
            \midrule
                20  & 21 &  210 &     1.5 \\
                25  & 26 &  325 &     3.3 \\
                30  & 31 &  465 &     7.7 \\
                35  & 36 &  630 &    20.8 \\
                40  & 41 &  820 &   106.9 \\
                45  & 46 & 1035 &   295.6 \\
                50  & 51 & 1275 &   181.6 \\
                55  & 56 & 1540 & timeout \\
            \bottomrule
            \addlinespace[0.5em]
            \toprule
                \multicolumn{4}{c}{Broadcast} \\
            \midrule
                & \multicolumn{1}{c}{$|Q|$} & \multicolumn{1}{c}{$|T|$} & \multicolumn{1}{c}{Time} \\
            \midrule
                & 2 & 1 & 0.1 \\
            \bottomrule
        \end{tabular}
    \end{minipage}%
    \hfill
    \begin{minipage}[t]{0.24\hsize}\centering
        \begin{tabular}[t]{rrrr}
            \toprule
                \multicolumn{4}{c}{Flock of birds~\cite{guidelines}} \\
            \midrule
                \multicolumn{1}{c}{$c$} & \multicolumn{1}{c}{$|Q|$} & \multicolumn{1}{c}{$|T|$}
                & \multicolumn{1}{c}{Time} \\
            \midrule
                 50  &  51 &  99 &    11.8 \\
                100  & 101 & 199 &    44.8 \\
                150  & 151 & 299 &   369.1 \\
                200  & 201 & 399 &   778.8 \\
                250  & 251 & 499 &  1554.2 \\
                300  & 301 & 599 &  2782.5 \\
                325  & 326 & 649 &  3470.8 \\
                350  & 351 & 699 & timeout \\
            \bottomrule
        \end{tabular}
    \end{minipage}%
}}}\caption{Results of the experimental evaluation where $|Q|$ denotes
  the number of states, $|T|$ denotes the number of non silent transitions, and
  the time to prove membership for \WSSS is given in seconds.}\label{tab:results}
\end{table}

\section{Experimental results}\label{sec:experimental}

We have developed a tool called \Peregrine
\footnote{\Peregrine and benchmarks 
are available from \url{https://gitlab.lrz.de/i7/peregrine/}.}
to check membership in \WSSS.
Peregrine is implemented on top of the
SMT solver Z3~\cite{DBLP:conf/tacas/MouraB08}. 

\Peregrine reads in a population protocol $\PP = (Q, T, \Sigma, I, O)$ 
and constructs two
sets of constraints. The first set is satisfiable if and only
if \layeredtermination holds, and the second is unsatisfiable if and
only if \snosplitterminal holds.

For \layeredtermination, our tool \Peregrine iteratively constructs
constraints checking the existence of an ordered 
partition of size $1,$$2,$ $\ldots,|T|$ 
and decides if they are satisfiable. To check that the execution
of a layer is silent, the constraints mentioned in the proof of
Proposition~\ref{prop:property1} are
transformed using Farkas' lemma (see \eg~\cite{Sch86}) into a version
that is satisfiable if and only if all the executions of the layer are 
silent. Also, the constraints for condition~(b) of
Definition~\ref{def:layers} are added.

For \snosplitterminal, \Peregrine initially constructs 
the constraints for the flow equation for three configurations 
$C_0, C_1, C_2$ and
vectors $\vec{x}_1$ and $\vec{x}_2$, with additional constraints
to guarantee that $C_0$ is initial, $C_1$ and $C_2$ are terminal, and
$C_1$ and $C_2$ are consensus of different values. If these
constraints are unsatisfiable, the protocol
satisfies \snosplitterminal. Otherwise, \Peregrine searches for a
$U$-trap or $U$-siphon to show that either $C_0 \ptrans{\vec{x}_1}
C_1$ or $C_0 \ptrans{\vec{x}_2} C_2$ does not hold.
If, say, a $U$-siphon
$\siphon$ is found, then \Peregrine adds the constraint $C_0(\siphon) > 0$ to
the set of initial constraints.  This process is repeated until either
the constraints are unsatisfiable and \snosplitterminal is shown, or
all possible $U$-traps and $U$-siphons are added, in which
case \snosplitterminal does not hold.  
We use this refinement-based approach instead of
the \coNP{} approach described in Proposition~\ref{strconscoNP},
as that could require a quadratic number of variables and constraints, and we
generally expect to need a small number of refinement steps.

We evaluated \Peregrine on a set of benchmarks: the threshold and
remainder protocols of~\cite{AADFP06}, the majority protocol
of~\cite{AAE06}, the broadcast protocol of~\cite{guidelines} and two
versions of the flock of birds\footnote{The variant
from~\cite{guidelines} is referred to as \emph{threshold-n} by its
authors.} protocol from~\cite{guidelines, CMS10}. We checked the
parametrized protocols for increasing values of their primary
parameter until we reached a timeout. For the threshold and remainder
protocols, we set the secondary parameter $c$ to $1$ since it has no
incidence on the size of the protocol, and since the variation in
execution time for different values of $c$ was negligible. Moreover,
we assumed that all possible values for $a_i$ were present in the
inputs, which represents the worst case.

All experiments were performed on the same machine equipped with an
Intel Core i7-4810MQ CPU and 16\,GB of RAM\@. The time limit was set
to 1 hour. The results are shown in Table~\ref{tab:results}.  In all
cases where we terminated within the time limit, we were able to show
membership for \WSSS. Generally, showing \snosplitterminal took much
less time than showing \layeredtermination, except for the
flock of birds protocols, where we needed linearly many $U$-traps.

As an extension, we also tried proving correctness
after proving membership in \WSSS.
For this, we constructed constraints
for the existence of an input $X$ and configuration $C$
with $I(X) \ptrans{\vec{x}} C$ and $\varphi(X) \neq O(C)$.
We were able to prove correctness for all the protocols
in our set of benchmarks. The correctness check was faster than the 
well-specification check for broadcast, majority, threshold and 
both flock of birds protocols, and slower for the remainder protocol,
where we reached a timeout for $m=70$.

\section{Conclusion and further work}\label{sec:conclusion}
We have presented \WSSS, the first class of well-specified population
protocols with a membership problem of reasonable complexity (\ie
in \DP) and with the full expressiveness of well-specified
protocols. Previous work had shown that the membership problem for the
general class of well-specified protocols is decidable, but at
least \EXPSPACE-hard with algorithms of non primitive recursive
complexity.

We have shown that \WSSS is a natural class that contains many
standard protocols from the literature, like flock-of-birds, majority,
threshold and remainder protocols. We implemented the membership
procedure for \WSSS on top of the SMT solver Z3, yielding the first
software able to automatically prove well-specification of population
protocols for \emph{all} (of the infinitely many) inputs. Previous
work could only prove partial correctness of protocols with at most 9
states and 28 transitions, by trying exhaustively a \emph{finite}
number of inputs~\cite{PLD08,SLDP09,CMS10,guidelines}. Our algorithm
deals with all inputs and can handle larger protocols with up to 70
states and over 2500 transitions.

Future work will concentrate on three problems: improving the
performance of our tool; automatically deciding if a \WSSS-protocol
computes the predicate described by a given Presburger formula; and
the diagnosis problem: when a protocol does not belong to \WSSS,
delivering an explanation, \eg a non-terminating fair execution. We
think that our constraint-based approach provides an excellent basis
for attacking these questions.


\clearpage

\bibliographystyle{acm}
\bibliography{references}

\clearpage
\appendix

\label{sec:appendix}
\section{Missing proofs of Section~\ref{sec:silent}}

For the proof of Proposition \ref{prop:silenthardness} we need to introduce 
Petri nets. Intuitively, Petri nets are similar to population protocols, but 
their transitions can also create and destroy agents.

A \emph{Petri net} $N = (P, T, F)$ consists of a finite set $P$ of
\emph{places}, a finite set $T$ of \emph{transitions}, and a {\em flow
  function} $F \colon (P \times T) \cup (T\times P) \rightarrow
\mathbb{N}$.  Given a transition \(t \in T\), the multiset \(\pre{t}\)
of \emph{input places} of $t$ is defined by \(\pre{t}(p) = F(p,t)\),
and the multiset \(\post{t}\) of \emph{output places} by \(\post{t}(p)
= F(t,p)\).  A \emph{marking} $M$ of a net $N$ is a multiset of
places. Given a place $p$, we say that $M$ puts $M(p)$ {\em tokens} in
$p$.  A transition $t\in T$ is \emph{enabled at} a marking $M$ if
\(\pre{t} \leq M\).  A transition $t$ enabled at $M$ can \emph{fire},
yielding the marking \(M' = M - \pre{t} + \post{t}\).  We write this
fact as $M\trans{t}M'$.  We extend enabledness and firing to sequences
of transitions as follows.  Let $\sigma=t_1\ldots t_k$ be a finite
sequence of transitions $t_j\in T$.  We write $M\trans{\sigma}M'$ and
call it a \emph{firing sequence} if there exists a sequence
$M_0,\ldots,M_k$ of markings such that
$M=M_0\trans{t_1}M_1\cdots\trans{t_k}M_k=M'$.  In that case, we say
that $M'$ is \emph{reachable from $M$} and denote by \({\it
  Reach}(N,M)\) the set of markings reachable from \(M\).

\propSilentHardness*

\begin{proof}
  The proof is very similar to the one of
  \cite[Theorem~10]{DBLP:conf/fsttcs/EsparzaGLM16}. However, since the
  proof requires small modifications at different places, we give it
  in full for completeness. The proof constructs a sequence of
  reductions from the Petri net reachability problem. Each step in the
  sequence transforms a problem on Petri nets into an equivalent
  problem closer to the model of population protocols. The first step
  uses a well-known result of Hack \cite{Hack76}. The reachability
  problem for Petri nets can be reduced in polynomial time to the
  single-place-zero-reachability problem:

  \begin{quote}
    Given a Petri net $N_0$, a marking $M_0$, and a place $\hat{p}$:
    decide whether some marking $M \in {\it Reach}(N_0, M_0)$
    satisfies $M(\hat{p})=0$.
  \end{quote}

  We introduce a normal form for Petri nets. A Petri net $N = (P, T,
  F)$ is said to be in \emph{normal form} if $F(x, y) \in \{0, 1\}$
  for every $x, y \in (P \times T) \cup (T \times P)$, and every
  transition $t$ satisfies $1 \leq |\pre{t}| \leq 2 $ and $1 \leq
  |\post{t}| \leq 2$. For every Petri net $N = (P, T, F)$ and markings
  $M_1, M_2$, one can construct a normal form Petri net
  $N'=(P',T',F')$ with $P \subseteq P'$ such that $M_2$ is reachable
  from $M_1$ in $N$ if and only if $M_2'$ is reachable from $M_1'$ in
  $N'$, and $M_i' = M_i + \vec{\ell}$, where $\ell$ is a special
  \emph{lock place}.
  Intuitively, each transition $t$ of $N$ with more than two input
  and/or output places is simulated in $N'$ by a \emph{widget}. The
  widget starts and finishes its execution by acquiring the lock and
  releasing it, respectively. This guarantees no two widgets are
  executing concurrently.  When simulating transition $t$, its
  widget first consumes, one by one, the tokens consumed by $t$ (as
  given by $\pre{t}$), and then produces, one by one, the tokens
  produced by $t$ (as given by $\post{t}$). Figure
  \ref{fig:widget} shows a transition and its widget.  Observe that
  all transitions of the widget are in normal form.

  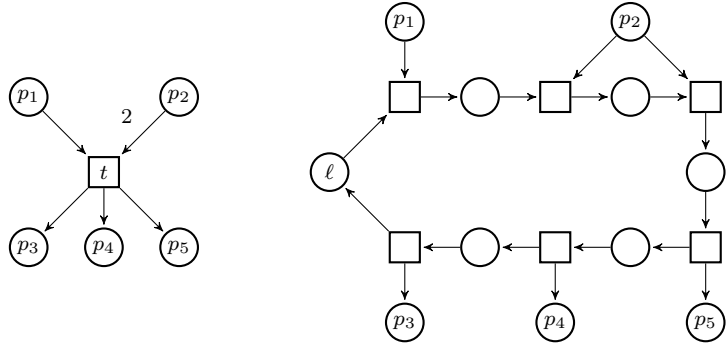
\begin{figure}[ht]
    \centering
    \scalebox{1}{\begin{tikzpicture}[>=stealth',bend angle=45,auto]
	\tikzstyle{every node}=[font=\scriptsize]
	\tikzstyle{place}=[circle,thick,draw=black,fill=white,minimum size=5mm,inner sep=0mm]
	\tikzstyle{transition}=[rectangle,thick,draw=black,fill=white,minimum size=4mm,inner sep=0mm]
	\tikzstyle{every label}=[black]
	
	\node [place] at (4,1) (p1){$p_1$};
        \node [place] at (6,1) (p2){$p_2$};
        \node [place] at (4,-1) (p4){$p_3$};
        \node [place] at (5,-1) (p5){$p_4$};
        \node [place] at (6,-1) (p6){$p_5$};

        \node [transition] at (5,0) (t1) {$t$}
  	edge [pre] (p1)  edge [pre] node {2} (p2) 
  	edge [post] (p4) edge [post] (p5) edge [post] (p6);

        \node [place] at (9,2) (np1){$p_1$};
        \node [place] at (12,2)(np2){$p_2$};
        \node [place] at (8,0) (r){$\ell$};
        \node [place] at (10,1)(p12){};
        \node [place] at (12,1)(p23){};
        \node [place] at (13,0)(p36){};
        \node [place] at (12,-1)(p65){};
        \node [place] at (10,-1) (p54){};
        \node [place] at (9,-2) (np4){$p_3$};
        \node [place] at (11,-2) (np5){$p_4$};
        \node [place] at (13,-2)(np6){$p_5$};

        \node [transition] at (9,1) (t1) {}
 	edge [pre] (r)   edge [pre] (np1)
        edge [post](p12);

        \node [transition] at (11,1) (t2) {}
 	edge [pre] (p12)   edge [pre] (np2)
        edge [post](p23);

        \node [transition] at (13,1) (t3) {}
 	edge [pre] (p23)   edge [pre] (np2)
        edge [post](p36);

        \node [transition] at (13,-1) (t6) {}
 	edge [pre] (p36) 
        edge [post](p65) edge [post](np6) ;

        \node [transition] at (11,-1) (t5) {}
 	edge [pre] (p65)   edge [post] (np5)
        edge [post](p54);

        \node [transition] at (9,-1) (t4) {}
 	edge [pre] (p54)   edge [post] (np4)
        edge [post](r);

\end{tikzpicture}}
    \caption{A Petri net transition (left) and its associated widget (right).}\label{fig:widget}
  \end{figure}
  
  Let $N_1$ be the result of normalizing $N_0$. Let $P_{aux}$ be the
  set of places of $N_1$ that are not places of $N_0$, and are
  different from the lock place $\ell$. The
  single-place-zero-reachability problem reduces to
  \begin{itemize}
  \item[(P1)] Does some marking $M \in {\it Reach}(N_1, M_0 +
    \vec{\ell})$ satisfy $M(\hat{p}) = 0$ and $M(P_{aux}) = 0$?
    \\ (Observe that $M(P_{aux}) = 0$ guarantees that no widget is in
    the middle of its execution.)
  \end{itemize}
  Now we add to $N_1$ a new place $p_0$ and a new widget simulating a
  transition $t_0$ with $ \pre{t_0} = \vec{p_0} $ and $\post{t_0} =
  M_0 + \vec{\ell}$.  Let the resulting net be $N_2$. Then (P1)
  reduces to:
  \begin{itemize}
  \item[(P2)] Does some marking $M \in {\it Reach}(N_2, \vec{p_0})$
    satisfy $M(\hat{p}) = M(p_0) = M(P_{aux}) = 0$?
  \end{itemize}

  For our next step, we ``reverse'' $N_2$: define $N_3$ as the result
  of reversing all arcs of ${N}_2$, \ie, $P_3 = P_2$, $T_3 = T_2$ but
  $F_3(x,y) = F_2(y,x)$ for every two nodes $x, y$.  Clearly, $N_3$ is
  in normal form when $N_2$ is. The problem~(P2) reduces to:

  \begin{itemize}
  \item[(P3)] Is $\vec{p_0} \in {\it Reach}(N_3, M)$ for some marking
    $M$ of $N_3$ satisfying $M(\hat{p}) = M(p_0) = M(P_{aux}) = 0$?
  \end{itemize}

  In the last step we reduce~(P3) to the membership problem for \WSS.
  Let $N_3 = (P_3, T_3, F_3)$. We construct a population protocol
  $\PP= (Q, T, \Sigma, I, O)$, defined as follows:
  \begin{itemize}
  \item $Q = P_3 \cup \{{\it Fresh}, {\it Used}, {\it
    Collect}\}$. That is, $\PP$ contains a state for each place of
    $N_3$, plus three auxiliary places.
  \item $\Sigma = Q \setminus (\{\hat{p}, p_0\} \cup P_{aux})$, and
    $I$ is the identity mapping.
  \item $O(p_0) = 1$ and $O(q) = 0$ for every $q \neq p_0$.
  \item $T = T_3' \cup T_P \cup T_s$.\\
  \end{itemize} 
These sets are formally described below. Intuitively, the transitions
of $T_3'$ simulate the Petri net transitions of $T_3$, the transitions
of $T_P$ guarantee that  a terminal consensus is reachable from every 
configuration that does not represent $\vec{p_0}$, 
and the $T_s$ are additional silent actions to make the protocol 
well-formed.

The transitions of $T_3'$ simulate the behaviour of $N_3$. For this,
$T_3'$ contains a transition $t'$ for every net transition $t \in
T_3$. If $t \in T_3$ has two input places $p_1, p_2$ and two output
places $p_1', p_2'$, then $t' = (p_1, p_2) \trans{} (p_1', p_2')$, The
other cases are: if $t$ has one input place $p_1$ and two output
places $p_1', p_2'$, then $t' = (p_1, {\it Fresh}) \trans{} (p_1',
p_2')$; if $t$ has two input places $p_1, p_2$ and one output place
$p_1'$, then $t' = (p_1, p_2) \trans{} (p_1', {\it Used})$; if $t$ has
one input place $p_1$ and one output place $p_1'$, then $t' = (p_1,
{\it Fresh}) \trans{} (p_1', {\it Used})$.

The transitions of $T_P$ test the presence of tokens anywhere, apart
from one single token in $p_0$. For every pair 
$(q, q') \in \left(\left(P_3 \setminus \set{p_0}\right) \times Q \right)  \cup \set{(p_0, p_0)}$, 
the set $T_P$ contains a transition
$(q, q') \mapsto ({\it Collect}, {\it Collect})$. Further, for every
place $q \in Q$, the set $T_P$ contains a transition $(q, {\it
  Collect}) \mapsto ({\it Collect}, {\it Collect})$.  Intuitively,
these transitions guarantee that as long as the current marking of
$N_3$ is different from $\mathbf{p}_0$, the protocol $\PP$ can reach a
terminal configuration with all agents in state ${\it Collect}$.

The set $T_s$ contains a silent transition $(q, q') \mapsto (q, q')$
for every pair $(q, q')$ of states.

Assume that $\vec{p_0} \in {\it Reach}(N_3, M)$ for some marking $M$
such that $M(\hat{p}) = M(p_0) = M(P_{aux}) = 0$. Let $\sigma$ be a
firing sequence such that $M \trans{\sigma} \vec{p_0}$. Observe that
$\sigma$ is nonempty, and must end with a firing of transition
$t_0$. Let $K$ be the number of times that transitions with only one
input place occur in $\sigma$. We claim that the initial configuration
$C$ given by $C({\it Fresh})= K$, and $C(p)=M(p)$ for every $p \in
P_3$ has a fair execution that does not reach a consensus. Indeed, the
finite execution of $\PP$ that simulates $\sigma$ by executing the
corresponding transitions of $T_3'$ (and which, abusing language, we
also denote $\sigma$), reaches a configuration $C'$ with $C'(p_0)=1$,
$C'({\it Fresh})=0$, $C'({\it Used}) > 0$ (because every transition that 
moves an agent to $p_0$ also moves an agent to ${\it Used}$), 
and $C'(p) = 0$ for any other place $p$. Since $O(p_0)= 1$ and $O({\it Used})=0$,
the configuration $C'$ is not a consensus configuration.
Since no transition of
$T_3' \cup T_P$ is enabled at $C'$, all transitions enabled at $C'$
are silent, and therefore from $C'$ it is not possible to reach a
consensus.

Assume now that $\vec{p_0} \notin {\it Reach}(N_3, M)$ for any marking
$M$ such that $M(\hat{p}) = M(p_0) = M(P_{aux}) = 0$. Then every
configuration reachable from any initial configuration enables some
transition of $T_P$. By fairness, every fair execution from any
initial configuration contains at least one transition of $T_P$, and
so some configuration reached along the execution populates state {\it
  Collect}. But then, again by fairness, the execution gets eventually
trapped in a terminal configuration $C$ of the form $C({\it Collect})>
0$ and $C(q)=0$ for every $q \notin {\it Collect}$. So every fair
execution is silent and stabilizes to $0$, and therefore the protocol
belongs to \WSS.

\end{proof}

\section{Missing proofs of Section~\ref{sec:strongly:silent}}

\begin{restatable}{proposition}{propWSS}\label{propWSS}
  A protocol belongs to \WSS if and only if it satisfies \termination
  and \nosplitterminal.
\end{restatable}

\begin{proof}
  We prove a stronger result:
  \begin{itemize}
  \item[(a)] A protocol is silent if and only if it satisfies
    \termination.
  \item[(b)] A silent protocol is well-specified if and only if it
    satisfies \nosplitterminal.
  \end{itemize}

  \noindent ((a) $\Rightarrow$): Follows immediately from the
  definitions. \medskip

  \noindent ((a) $\Leftarrow$): Let $C_0$ be an arbitrary configuration,
  and let $\gamma = C_0 C_1 C_2 \cdots$ be a fair execution of the
  protocol. Let $\mathcal{C}_\bot$ be the set of terminal
  configurations reachable from $C_0$. Since \termination holds for
  every configuration, and so in particular for all of $C_1, C_2,
  \ldots$, all configurations of $\gamma$ can reach some cofiguration
  of $\mathcal{C}_\bot$.

  For every $C_i$, let $d(C_i)$ be the length of a shortest path from
  $C_i$ to some configuration of $\mathcal{C}_\bot$. We claim that for
  every $n \geq 0$, there are infinitely many indices $i$ such that
  $d(C_i) \leq n$. Since there are only finitely many configurations
  reachable from $C_0$, say $K$, we have $d(C_i) \leq K$ for every
  index $i \geq 0$. So it suffices to show that if there are
  infinitely many indices $i$ such that $d(C_i) \leq n$, then there
  are infinitely many indices $j$ such that $d(C_j) \leq n-1$.

  Let $i_1 \leq i_2 \leq i_3 \cdots$ be an infinite collection of
  indices such that $d(C_{i_j}) \leq n$ for every $j \geq 1$. By
  definition of $d$, for every configuration $C_{i_j}$ there is a
  step $C_{i_j}\trans{}C_{i_j}'$ such that $d(C_{i_j}')=n-1$.  By
  fairness, we have $C_{i_j}' = C_{i_j +1}$ for infinitely many $j
  \geq 1$, and the claim is proved. By this claim, there are
  infinitely many indices $i$ such that $d(C_i) \leq 0$, \ie, $C_i \in
  \mathcal{C}_\bot$. Let $i_0$ be one of them. Since
  $\mathcal{C}_\bot$ only contains terminal configurations, we have
  $C_i = C_{i_0}$ for every $i \geq i_0$, and so $\gamma$ converges to
  $C_\bot$.

  \medskip \noindent ((b) $\Rightarrow$) Let $\PP$ be a silent and
  well-specified protocol. Let $C_0$ be an initial configuration of
  $\PP$, and let $C_0 C_1 \cdots C_n$ be a finite prefix of an
  execution such that $C_n$ is terminal. The execution $C_0 C_1 \cdots
  (C_n)^\omega$ is fair. Since the protocol is well-specified, $C_n$
  is a consensus configuration.

  \medskip \noindent ((b) $\Leftarrow$) Let $\PP$ be a silent protocol
  satisfying \nosplitterminal. By convergence, every fair execution
  starting at an initial configuration $C$ eventually reaches a
  terminal configuration. Since $\PP$ satisfies \nosplitterminal, all
  these configurations are consensus configurations, and moreover they
  all agree to the same boolean value.
\end{proof}

\begin{restatable}{proposition}{layterimpter}\label{layterimpter}
  \layeredtermination implies \termination.
\end{restatable}

\begin{proof}
  Let $\PP = (Q, T, \Sigma, I, O)$ be a population protocol satisfying
  \layeredtermination, and let $C$ be an arbitrary configuration of
  $\PP$. Let $(T_1, T_2, \ldots, T_n)$ be the ordered partition of $T$
  for \layeredtermination. There exists a sequence $w_1 \in T_1^*$
  such that $C \trans{w_1} C_1$, and $C_1$ is a terminal configuration
  of $\PP[T_1]$. By the same reasoning, there exists a sequence $w_2
  \in T_2^*$ such that $C_1 \trans{w_2} C_2$, and $C_2$ is a terminal
  configuration of $\PP[T_2]$; further, by condition~(b) of Definition~\ref{def:layers},
  $C_2$ is also a terminal configuration of $T_1$. Iterating this process we
  find $C_1 \trans{w_1 \ldots w_n} C_n$ such that $C_n$ is a terminal
  configuration of all of $\PP[T_1], \PP[T_2], \ldots,
  \PP[T_n]$. Therefore, $C_n$ is a terminal configuration of $\PP$.
\end{proof}

We prove the claim made in the proof of
Proposition~\ref{prop:property1}:

\begin{restatable}{lemma}{tinv}\label{lem:tinv}
  Let $\PP = (Q, T, \Sigma, I, O)$ be a population protocol and let
  $U$ be the set of non silent transitions of $T$. $\PP$ has a
  non silent execution if and only if there exists $\vec{x} : U \to
  \Q$ such that $\sum_{t \in U} \vec{x}(t) \cdot (\post{t}(q) -
  \pre{t}(q)) = 0$ and $\vec{x}(q) \geq 0$ for every $q \in Q$, and
  $\vec{x}(q) > 0$ for some $q \in Q$.
\end{restatable}

\begin{proof}
  $\Rightarrow$) Let $C_0 C_1 \cdots$ be a non silent execution. Since
  taking a silent transition does not change the current
  configuration, we can assume that the transitions occurring in the
  execution are non silent. Since the total number of agents of a
  configuration is left unchanged by transitions, there exist indices
  $i < j$ such that $C_i = C_j$. Let $w$ be the non empty sequence of
  transitions leading from $C_i$ to itself. By the flow equation with
  $C \defeq C_i$ and $C' \defeq C_i$ we have $\sum_{t \in T} |w|_t
  \cdot (\post{t}(q) - \pre{t}(q)) = 0$ for every state $q$. Define
  $\vec{x}(t) \defeq |w|_t$.\medskip

  $\Leftarrow$) Without loss of generality, we can assume $\vec{x}(q)
  \in \N$. If this is not the case, we multiply $\vec{x}$ by a
  suitable coefficient which yields another vector satisfying the
  conditions. Let $w \in U^*$ be any sequence of transitions such that
  $|w|_t = \vec{x}(t)$ for every $t \in U$, and $|w|_t = 0$
  otherwise. Choose a configuration $C$ such that $C \trans{w} C'$ for
  some configuration $C'$. Observe that $C$ exists, for example it
  suffices to take $C(q) > 2 \cdot |w|$ for every state $q$. By the
  flow equation, we have $C' = C$. So $w^\omega$ is a non silent
  execution of $\PP$ from $C$.
\end{proof}

We prove the claim made in the proof of
Proposition~\ref{prop:property2}:

\begin{restatable}{lemma}{udead}\label{lem:udead}
  Let $\PP= (Q, T, \Sigma, I, O)$ be a protocol, and let $U \subseteq
  T$. It can be decided in polynomial time if $\PP$ is $U$-dead.
\end{restatable}

\begin{proof}
  Let $S = T \setminus U$. We claim that $\PP$ is \emph{not} $U$-dead
  if and only if there exist $s \in S$ and non silent $u \in U$ such that
  \begin{align}
    \pre{u'} \nleq \pre{s} + (\pre{u} \mminus \post{s}) \text{ for
      every } u' \in U \text{ such that } \post{u'} \neq
    \pre{u'}.\label{eq:t:dead}
  \end{align}
  The polynomial time upper bound follows by simply testing this
  condition on every pair $(s, u) \in S \times U$. Let us now prove
  the claim.

  $\Leftarrow$) Suppose there exist $s \in S$ and non silent $u \in U$ such
  that~(\ref{eq:t:dead}) holds. Let $C_0 \in \pop{Q}$ be the
  configuration such that $C_0 \defeq \pre{s} +
  (\pre{u} \mminus \post{s})$. By~(\ref{eq:t:dead}), $C_0$ 
  is $U$-dead. Moreover, $C_0 \trans{s} C$ where $C =
  (\pre{u} \mminus \post{s}) + \post{s}$. Therefore, $u$ can be taken
  from $C$ since $\pre{u} \leq C$. We conclude that $\PP$ is not
  $U$-dead from $C_0$.

  $\Rightarrow$) If $\PP$ is not $U$-dead, then there exists $C_0
  \trans{s_1} C_1 \trans{s_2} \cdots \trans{s_n} C_n$ such that $C_0,
  \ldots, C_{n-1}$ are $U$-dead, $s_1, s_2, \dots, s_n \in S$, and
  $C_n$ is not $U$-dead. Let $u \in U$ be some non silent transition that can be
  taken from $C_n$, \ie such that $\pre{u} \leq C_n$. Suppose there
  exists some $u' \in U$ such that $\post{u'} \neq \pre{u'}$ and
  $\pre{u'} \leq \pre{s_n} + (\pre{u} \mminus \post{s_n})$. We obtain
  \begin{align*}
    \pre{u'}
    &\leq \pre{s_n} + (\pre{u} \mminus \post{s_n}) \\[2pt]
    &\leq \pre{s_n} + (C_n \mminus \post{s_n}) && \text{(by $\pre{u} \leq C_n$)} \\[2pt]
    &\leq C_n \mminus \post{s_n} + \pre{s_n} \\
    &= C_{n-1} && \text{(by $C_{n-1} \trans{s_n} C_n$)\ .}
  \end{align*}
  Therefore, $C_{n-1} \trans{u'} C$ for some $C \in
  \pop{Q}$. Moreover, $C \neq C_{n-1}$ since $\post{u'} \neq
  \pre{u'}$. This contradicts the fact that $C_{n-1}$ is $U$-dead,
  hence~(\ref{eq:t:dead}) holds.
\end{proof}

\section{Missing proofs of Section~\ref{sec:strong:protocols}}

\subsection{Threshold protocol}

We first prove the two claims made in the proof of
Prop.~\ref{prop:threshold:strong:consensus} and then give a full proof
of Prop.~\ref{prop:threshold:strong:consensus}.

\begin{proposition}\label{prop:threshold:val}
  For every $C, C' \in \pop{Q}$ and $\vec{x} : T \to \N$, if $(C, C',
  \vec{x})$ is a solution to the flow equations, then $\val{C} =
  \val{C'}$.
\end{proposition}

\begin{proof}
  Assume $(C, C', \vec{x})$ is a solution to the flow equations. For
  every $m, n \in [-\lmax, \lmax]$, we have $g(m, n) + f(m, n) = m +
  n$. Therefore, $\val{\pre{t}} = \val{\post{t}}$ for every $t \in
  T$. This implies:
  \begin{align*}
      \val{C'} &= \sum_{q \in Q} (C(q) + \sum_{t \in T} \vec{x}(t) \cdot (\post{t}(q) -
        \pre{t}(q))) \cdot \val{q} \\
    &= \val{C} + \sum_{q \in Q} \sum_{t \in T} \vec{x}(t) \cdot (\post{t}(q) -
    \pre{t}(q)) \cdot \val{q} \\
    &= \val{C} + \sum_{t \in T} \vec{x}(t) \cdot \left[\sum_{q \in Q}
      \post{t}(q) \cdot \val{q} - \sum_{q \in Q} \pre{t}(q) \cdot
      \val{q}\right] \\
    &= \val{C} + \sum_{t \in T} \vec{x}(t) \cdot \val{\post{t}} -
    \val{\pre{t}} \\
    &= \val{C}. \qedhere
  \end{align*}
\end{proof}

\begin{proposition}\label{prop:threshold:term}
  Let $C, C' \in \pop{Q}$ be terminal configuration that contain a
  leader. Both $C$ and $C'$ are consensus configurations. Moreover, if
  $\val{C} = \val{C'}$, then $O(C) = O(C')$.
\end{proposition}

\begin{proof}
  We prove the first claim for $C$. The argument is identical for
  $C'$. Suppose that $C$ is not a consensus configuration. Let $(1, m,
  o) \in \supp{C}$ be a leader of $C$. Since $C$ is not a consensus
  configuration, there exists $(\ell, n, \neg o) \in
  \supp{C}$. Therefore, the following transition $t$ is enabled at
  $C$:
  \begin{align*}
    (1, m, o), (\ell, n, \neg o) \mapsto (1, f(m, n), b(m, n)), (0,
    g(m, n), b(m, n)).
  \end{align*}
  Moreover, $t$ is non silent which contradicts the fact that $C$ is
  terminal. Thus, $C$ is a consensus configuration.

  Assume that $\val{C} = \val{C'}$. Suppose that $O(C) \neq O(C')$ for
  the sake of contradiction. Without loss of generality, we may assume
  that $O(C) = 1$ and $O(C') = 0$. Let $p_C, p_{C'} \in Q$ be
  respectively leaders of $C$ and $C'$. We have $\val{p_C} < c <
  \lmax$ and $\val{p_{C'}} \geq c > -\lmax$. We claim that
  \begin{align}
    \val{p_C} \geq \val{C} \text{ and } \val{p_{C'}} \leq
    \val{C'}.\label{eq:val:pc}
  \end{align}
  To see that the claim holds, suppose that $\val{p_C} <
  \val{C}$. There exists some $q_C \in \supp{C}$ such that $\val{q_C}
  > 0$. Since $\val{p_C} < \lmax$, some part of the value of $q_C$ can
  be transferred to $p_C$, \ie there exists a non silent transition $t
  \in T$ with $\pre{t} = \multiset{p_C, q_C}$, which contradicts that
  $C$ is terminal. Thus, $\val{p_C} \geq \val{C}$ holds. The case
  $\val{p_{C'}} \leq \val{C'}$ follows by a similar argument.

  Now, by~(\ref{eq:val:pc}) we have $ \val{C} \leq \val{p_C} < c$ and
  $\val{C'} \geq \val{p_{C'}} \geq c$ which is a contradiction since
  $\val{C} = \val{C'}$. Therefore, $O(C) = O(C')$.\qedhere
\end{proof}

\propThresholdConsensus*

\begin{proof}
  Let $\val{q} \defeq n$ for every state $q = (\ell, n, o) \in Q$, and
  let $\val{C} \defeq \sum_{q \in Q} C(q) \cdot \val{q}$ for every
  configuration $C \in \pop{Q}$. Suppose for the sake of contradiction
  that $\PP_{\text{thr}}$ does not satisfy \snosplitterminal. There
  are two cases to consider.
  \begin{itemize}
    \item There exist $C, C' \in \pop{Q}$ such that $C \ptrans{} C'$,
      $C$ is initial, $C'$ is terminal and $C'$ is not a consensus
      configuration. Since $C$ is initial, it contains a leader. It is
      readily seen that the set of leaders forms a $U$-trap for every
      $U \subseteq T$, which implies that $C'$ contains a leader as
      $(C, C', \vec{x})$ satisfying $U$-trap constraints for all $U$.
      By Prop.~\ref{prop:threshold:term}, $C'$ is a consensus
      configuration, which is a contradiction.\medskip

    \item There exist $C_0, C, C' \in \pop{Q}$ and $\vec{x}, \vec{x}'
      : T \to \N$ such that $C_0 \ptrans{\vec{x}} C$, $C_0
      \ptrans{\vec{x}'} C'$, $C_0$ is initial, $C$ and $C'$ are
      terminal consensus configurations, and $O(C) \neq O(C')$. Note
      that $(C_0, C, \vec{x})$ and $(C_0, C', \vec{x}')$ both satisfy
      the flow equations. Therefore, by
      Prop.~\ref{prop:threshold:val}, $\val{C} = \val{C_0} =
      \val{C'}$. Again, since $C_0$ is initial, it contains a leader,
      which implies that both $C$ and $C'$ contain a leader. Since
      $\val{C} = \val{C'}$, Prop.~\ref{prop:threshold:term} yields
      $O(C) = O(C')$ which is a contradiction.\qedhere
  \end{itemize}
\end{proof}

We now give a full proof of Proposition~\ref{prop:threshold:layered}.

\propThresholdLayered*

\renewcommand{\leaders}[1]{\mathrm{leaders}(#1)}
\newcommand{\passive}[1]{\mathrm{nonleaders}(#1)}
\newcommand{\numleaders}[1]{\mathrm{num}\text{-}\mathrm{leaders}(#1)}
\newcommand{\numpassive}[1]{\mathrm{num}\text{-}\mathrm{nonleaders}(#1)}

\begin{proof}
  Assume $c > 0$. The case where $c \leq 0$ follows by a symmetric
  argument. Let $L_0 \defeq \{(1, x, 0) : c \leq x \leq \lmax\}$ and
  $N_1 \defeq \{(0, 0, 1)\}$. We claim that the following ordered
  partition satisfies layered termination:
  \begin{align*}
    T_1 &\defeq \{t \in T : \pre{t} \neq \multiset{q, r} \text{ for all } q \in L_0, r \in N_1\}, \\ T_2
    &\defeq T \setminus T_1.
  \end{align*}
  We first show that every execution of $\PP_{\text{thr}}[T_1]$ is
  fair. For the sake of contradiction, assume this is not the
  case. There exists a non silent execution $C_0 \trans{t_1} C_1 \trans{t_2}
  \cdots$ where $C_0, C_1, \ldots \in \pop{Q}$ and $t_1, t_2, \ldots
  \in T_1$. For every $i \in \N$, let
  \begin{align*}
    \leaders{C_i} &\defeq \{q \in \supp{C_i} : q \text{ is a leader}\},\\
    \passive{C_i} &\defeq \{q \in \supp{C_i} : q \text{ is not a
      leader}\}, \\
    \numleaders{C_i} &\defeq \sum_{q \in \leaders{C_i}} C_i(q), \\
    a_i &\defeq \sum_{\passive{C_i}} C_i(q) \cdot |\val{q}|.
  \end{align*}
  As observed in~\cite{AADFP04}, $a_0 \geq a_1 \geq
  \cdots$. Therefore, there exists $n_1 \in \N$ such that $a_i =
  a_{i-1}$ for every $i > n_1$. Similarly, no transition of
  $\PP_{\text{thr}}$ increases the number of leaders. Thus, there
  exists $n_2 \in \N$ such that $\numleaders{C_i} =
  \numleaders{C_{i-1}}$ for every $i > n_2$. Let $L_\text{err} \defeq
  \{(1, x, b) : -\lmax \leq x \leq \lmax, b \not= (x < c)\}$ be the
  set of leaders whose opinion is inconsistent with their value. Since
  no transition of $\PP_{\text{thr}}$ produces states from
  $L_\text{err}$, transitions involving a state from $L_\text{err}$ can
  only be taken in finitely many steps. More formally, there exists
  $n_3 \in \N$ such that $\supp{\pre{t_i}} \cap L_\text{err} =
  \emptyset$ for every $i > n_3$. Let $n \defeq \max(n_1, n_2,
  n_3)$. Any non silent transition $t_i$ such that $i > n$ must be of
  the form:
  \begin{align*}
    (1, x, 1), (0, 0, 0) &\mapsto (1, x, 1), (0, 0, 1)
  \end{align*}
  for some $x < c$, as otherwise one of the above observations would
  be violated. But such transitions set the opinion of non leaders to
  $0$, which can only occur for finitely many steps. Therefore, there
  exists $n' \geq n$ such that every transition enabled in $C_{n'}$ is
  silent. This is a contradiction.

  It is readily seen that any execution of $\PP_{\text{thr}}[T_2]$ is
  silent since each transition of $T_2$ is of the form:
  \begin{align*}
    (1, x, 0), (0, 0, 1) \mapsto (1, x, 0), (0, 0, 0)
  \end{align*}
  for some $c \leq x \leq \lmax$. Therefore, it remains to prove that
  $\PP_{\text{thr}}[T_2]$ is $T_1$-dead. Let $C \in \pop{Q}$ be a
  $T_1$-dead configuration. For the sake of contradiction, suppose
  there exists $w \in T_2^+$ and $C' \in \pop{Q}$ such that $C
  \trans{w} C'$ and $C'$ enables some non silent transition $t \in
  T_1$. Since $C$ is $T_1$-dead, transition $t$ must be of the form
  \begin{align*}
    (1, y, 1), (0, 0, 0) \mapsto (1, y, 1), (0, 0, 1)
  \end{align*}
  for some $y < c$. Moreover, $(1, y, 1)$ already appeared in
  $C$. This means that $C$ contains one leader of opinion $0$, and one
  leader of opinion $1$. Therefore, $C$ is not $T_1$-dead, which is a
  contradiction.
\end{proof}

\subsection{Remainder protocol}

Let $\val{C} \defeq \left( \sum_{n \in [0, m)} C(n) \cdot n \right) \text{ mod } m$ for every $C \in
  \pop{Q}$.

\begin{proposition}
  $\PP_{\text{rmd}}$ satisfies \snosplitterminal.
\end{proposition}

\begin{proof}
  For every $C', C' \in \pop{Q}$,
  \begin{itemize}
  \item[(a)] if $(C, C', \vec{x})$ is a solution to the flow equations
    for some $\vec{x} : T \to \N$, then $\val{C} = \val{C'}$.
    
  \item[(b)] if $C, C' \in \pop{Q}$ are terminal configuration that contain a
    numerical value, then both $C$ and $C'$ are consensus
    configurations, and if $\val{C} = \val{C'}$, then $O(C) = O(C')$.
  \end{itemize}
  The proof of these two claims follows from the definition of
  $\PP_{\text{rmd}}$ as in the case of the threshold protocol.

  Suppose for the sake of contradiction that $\PP_{\text{rmd}}$ does
  not satisfy \snosplitterminal. There are two cases to consider.

  \begin{itemize}
    \item There exist $C, C' \in \pop{Q}$ such that $C \ptrans{} C'$,
      $C$ is initial, $C'$ is terminal and $C'$ is not a consensus
      configuration. Since $C_0$ is initial, it only contains
      numerical values. Since numerical values form a $U$-trap for
      every $U \subseteq T$, $C$ contains a numerical value. By~(b), $C$ is
      a consensus configuration, which is a contradiction.\medskip

    \item There exist $C_0, C, C' \in \pop{Q}$ and $\vec{x}, \vec{x}'
      : T \to \N$ such that $C_0 \ptrans{\vec{x}} C$, $C_0
      \ptrans{\vec{x}'} C'$, $C_0$ is initial, $C$ and $C'$ are
      terminal consensus configurations, and $O(C) \neq O(C')$. Note
      that $(C_0, C, \vec{x})$ and $(C_0, C', \vec{x}')$ both satisfy
      the flow equations. Therefore, by~(a), $\val{C} =
      \val{C_0} = \val{C'}$. Again, since $C_0$ is initial, it
      contains a numerical value, which implies that both $C$ and $C'$
      contain a numerical value. Since $\val{C} = \val{C'}$,
      (b)~yields $O(C) = O(C')$ which is a contradiction.\qedhere
  \end{itemize}
\end{proof}

\begin{proposition}
  $\PP_{\text{rmd}}$ satisfies \layeredtermination.
\end{proposition}

\newcommand{\numerical}[1]{\mathrm{numerical}(#1)}

\begin{proof}
  We claim that the following ordered partition satisfies layered
  termination:
  \begin{align*}
    T_1 &\defeq \{t \in T : \pre{t} = \multiset{q, r} \text{ for some } q \in [0, m), r \in ([0, m) \cup
        \{\false\})\} \\
    T_2 &\defeq \{t \in T : \pre{t} = \multiset{q, \true} \text{ for some } q \in [0, m)\}
  \end{align*}
  We first show that every execution of $\PP_{\text{rmd}}[T_1]$ is
  silent. For the sake of contradiction, assume it is not the
  case. There exists a non silent execution $C_0 \trans{t_1} C_1 \trans{t_2}
  \cdots$ where $C_0, C_1, \ldots \in \pop{Q}$ and $t_1, t_2, \ldots
  \in T_1$. For every $i \in \N$, let $\numerical{C_i} \defeq \sum_{n
    \in [0, m)} C_i(n)$. It is readily seen that $\numerical{C_0} \geq
    \numerical{C_1} \geq \cdots$. Therefore, there exists $\ell \in
    \N$ such that $\numerical{C_i} = \numerical{C_{i-1}}$ for every $i
    > \ell$. This implies that, for every $i > \ell$, if $t_i$ is non
    silent, then it is of the form $(n, \false) \mapsto (n, \true)$
    for some $n \in [0, m)$. But, these non silent transitions can
      only occur for a finite amount of steps, which is a
      contradiction.

  It is readily seen that every execution of $\PP_{\text{rmd}}[T_2]$
  is silent since non silent transitions of $T_2$ are all of the form
  $(n, \true) \mapsto (n, \false)$ for some $n \in [0, m)$. Therefore,
    it remains to prove that $\PP_{\text{rmd}}[T_2]$ is
    $T_1$-dead. Let $C \in \pop{Q}$ be a $T_1$-dead configuration. For
    the sake of contradiction, suppose there exists $w \in T_2^+$ and
    $C' \in \pop{Q}$ such that $C \trans{w} C'$ and $C'$ enables some
    non silent transition $t \in T_1$. We have $C(\true) > 0$ and
    $C(n) > 0$ for some $n \in [0, m)$ such that $O(n) =
      \false$. Moreover, since $C$ is $T_1$-dead, $\numerical{C} =
      1$. Therefore $t$ must be of the form $(n, \false) \mapsto (n,
      \false)$. We obtain a contradiction since $t$ is non silent.
\end{proof}

\subsection{Conjunction protocol}

For the rest of this subsection, let us fix some population protocols
$\PP_1 = (Q_1, T_1, \Sigma, I_1, O_1)$ and $\PP_2 = (Q_2, T_2, \Sigma,
I_2, O_2)$. For every transition transition 
$t = (q, r) \mapsto (q', r')$, and every pair of states $(p, s)$, 
let $t \ziptimes (p,s)$ denote the transition lifted to $(p, s)$:
\begin{align*}
  ((q, p), (r, s)) \mapsto ((q', p), (r', s))
\end{align*}
Similarly let 
$(p, s) \ziptimes t$ denote the lifted transition 
\begin{align*}
((p, q), (s, r)) \mapsto ((p, q'), (s, r')).
\end{align*}

\begin{definition}
  The \emph{conjunction} of $\PP_1$ and $\PP_2$ is defined as the
  population protocol $\PP \defeq (Q, S, I, \Sigma, O)$ such that $Q
  \defeq Q_1 \times Q_2$, $S \defeq S_1 \cup S_2$, $I(\sigma) \defeq
  (I_1(\sigma), I_2(\sigma))$ and $O(p, q) \defeq O_1(p) \land O_2(q)$
  where
  \begin{align*} S_1 &\defeq \{ t \ziptimes (q, r): t \in T_1, (q, r) \in Q_2\times Q_2\},
    \\ S_2 &\defeq \{(q, r) \ziptimes t : t \in T_2, (q, r) \in Q_1 \times Q_1\}.
  \end{align*}
\end{definition}

The \emph{projection} of $q \in Q$ on $Q_i$ is the state $\pi_i(q)
\defeq q_i$ where $q = (q_1, q_2)$. The \emph{projection} of $t \in
S_i$ on $T_i$ is the transition $\pi_i(t) \defeq (\pi_i(p), \pi_i(q),
\pi_i(p'), \pi_i(q'))$ where $t = (p, q, p', q')$. We lift projections
to $\pop{Q}$ and $S \to \N$ as follows. For every $C \in \pop{Q}$ and
$\vec{x} : S \to \N$, the
\emph{projections} $\pi_i(C) \in \pop{Q_i}$ and $\pi_i(\vec{x})
: T_i \to \N$ are respectively the configuration and mapping such that
\begin{align*}
  \pi_i(C)(q) \defeq \sum_{\substack{r \in Q \\ \pi_i(r) = q}}
  C(r) \text{ for every } q \in Q_i\quad \text{ and }\quad
  \pi_i(\vec{x})(t) &\defeq \sum_{\substack{s \in S_i \\ \pi_i(s) =
      t}} \vec{x}(s) \text{ for every } t \in T_i.
\end{align*}

Let $\incid{\PP} \in \N^{Q \times T}$ be the matrix such that
$\incid{\PP}(q, t) \defeq \post{t}(q) - \pre{t}(q)$ for every $q \in
Q$ and $t \in T$. It is readily seen that $(C, C', \vec{x})$ satisfies the flow
equations if and only if $C' = C + \incid{\PP} \cdot \vec{x}$. The
same holds for the matrices $\incid{\PP_1}$ and $\incid{\PP_1}$
defined similarly for $\PP_1$ and $\PP_2$. The following holds:

\begin{proposition}\label{prop:conj:distrib}
  For every $i \in \{1, 2\}$, $C, C' \in \pop{Q}$ and $\vec{x} \in S
  \to \N$ we have:
  \begin{itemize}
   \item[(a)] $\pi_i(C + C') = \pi_i(C) + \pi_i(C')$, and
   \item[(b)] $\pi_i\left(\incid{\PP} \cdot \vec{x} \right) =
     \incid{\PP_i} \cdot \pi_i(\vec{x})$.
  \end{itemize}
\end{proposition}

\begin{proof}
  For every $q \in Q$, we have
  \begin{align*}
    \pi_i(C + C')(q) &= \sum_{\substack{r \in Q \\ \pi_i(r) = q}}
    (C + C')(r) && \text{(by def. of $\pi_i$)} \\
    &= \sum_{\substack{r \in Q \\ \pi_i(r) = q}} C(r) + C'(r) \\
    &= \sum_{\substack{r \in Q \\ \pi_i(r) = q}} C(r) +
    \sum_{\substack{r \in Q \\ \pi_i(r) = q}} C'(r) \\
    &= \pi_i(C) + \pi_i(C') && \text{(by def. of $\pi_i$)}.
  \end{align*}
  This proves~(a). Let us now prove~(b). Let $i \in \{1, 2\}$ and $q
  \in Q_i$. By definition of $S$, we have
  \begin{align}
    \sum_{\substack{r \in Q \\ \pi_i(r) = q}} \incid{\PP}(r, t) &= 0
    &&\text{ for every } t \in S \setminus S_i,\label{eq:incid:zero} \\
    \sum_{\substack{r \in Q \\ \pi_i(r) = q}} \incid{\PP}(r, t) &=
    \incid{\PP_i}(q, \pi_i(t))
    &&\text{ for every } t \in S_i.\label{eq:incid:sub}
  \end{align}
  Therefore,
  \begin{align*}
    \pi_i(\incid{\PP} \cdot \vec{x})(q) &= \sum_{\substack{r \in Q
        \\ \pi_i(r) = q}} (\incid{\PP} \cdot \vec{x})(r) && \text{(by
      def. of $\pi_i$)} \\
    &= \sum_{\substack{r \in Q \\ \pi_i(r) = q}} \sum_{s \in S}
    \incid{\PP}(r, s) \cdot \vec{x}(s) \\
    &= \sum_{s \in S} \vec{x}(s) \cdot \sum_{\substack{r \in Q
        \\ \pi_i(s) = q}} \incid{\PP}(r, s) \\
    &= \sum_{s \in S_i} \vec{x}(s) \cdot \sum_{\substack{r \in Q
        \\ \pi_i(s) = q}} \incid{\PP}(r, s) &&
    \text{(by~(\ref{eq:incid:zero}))} \\
    &= \sum_{s \in S_i} \vec{x}(s) \cdot \incid{\PP_i}(q, \pi_i(s)) &&
    \text{(by~(\ref{eq:incid:sub}))} \\
    &= \sum_{t \in T_i} \incid{\PP_i}(q, t) \cdot \sum_{\substack{s
        \in S_i \\ \pi_i(s) = t}} \vec{x}(s) \\
    &= \sum_{t \in T_i} \incid{\PP_i}(q, t) \cdot \pi_i(\vec{x})(t) &&
    \text{(by def. of $\pi_i$)} \\
    &= (\incid{\PP_i} \cdot \pi_i(\vec{x}))(q). \tag*{\qedhere}
  \end{align*}
\end{proof}

\begin{restatable}{proposition}{propConjPtrans}\label{prop:conj:ptrans}
  For every $C, C' \in \pop{Q}$, $\vec{x} : S \to \N$ and $i \in \{1,
  2\}$, if $C \ptrans{\vec{x}} C'$, then
  $\pi_i(C) \ptrans{\pi_i(\vec{x})} \pi_i(C')$.
\end{restatable}

\begin{proof}
  \noindent\emph{Flow equations:} We have $C' = C + \incid{\PP} \cdot
  \vec{x}_j$. Therefore, for every $i \in \{1, 2\}$,
  \begin{align*}
    \pi_i(C')
    &= \pi_i(C + \incid{\PP} \cdot \vec{x}) \\
    &= \pi_i(C) + \pi_i(\incid{\PP} \cdot \vec{x}) && \text{
      (by Proposition~\ref{prop:conj:distrib}(a))} \\
    &= \pi_i(C) + \incid{\PP_i} \cdot \pi_i(\vec{x}) && \text{ (by
        Proposition~\ref{prop:conj:distrib}(b))}.
  \end{align*}

  \noindent\emph{Trap constraints:} For the sake of contradiction,
  suppose there exists $i \in \{1, 2\}$ such that a $U$-trap constraint is
  violated by $(\pi_i(C), \pi_i(C'), \pi_i(\vec{x}))$ for some $P
  \subseteq Q_i$. As both cases are symmetric, we may assume without
  loss of generality that $i = 1$. We have
  \begin{align}
    \preset{P} \cap \supp{\pi_1(\vec{x})} \neq \emptyset,\ \postset{P} \cap
    \supp{\pi_1(\vec{x})} \subseteq \preset{P}\ \text{ and
    }\ C'(P) = 0\label{eq:trap:constraint}
  \end{align}
  Let $R \defeq P \times Q_2$. By definition of projections, we have
  \begin{align}
    \pi_1(C')(P) = 0 &\iff C'(R) = 0.
  \end{align}
  We claim that
  \begin{align}
    \preset{R} \cap \supp{\vec{x}} &\neq \emptyset, \label{eq:claim:trap_nonempty} \\
    \postset{R} \cap \supp{\vec{x}} &\subseteq
    \preset{R}.\label{eq:claim:trap}
  \end{align}
  If~(\ref{eq:claim:trap_nonempty}) and~(\ref{eq:claim:trap}) hold,
  then we are done since the $\supp{\vec{x}}$-trap constraint for $R$
  is violated. So, it remains to prove these claims. To
  prove~(\ref{eq:claim:trap_nonempty}), let $t \in \preset{P} \cap
  \supp{\pi_1(\vec{x})}$. By assumption, such a $t$ must exist.  Since
  $t \in \preset{P}$, we have that $t \colon (p, p') \mapsto (r, r')$
  with $r \in P$ or $r' \in P$.  Moreover, since $t \in
  \supp{\pi_1(\vec{x})}$, by definition of projections there must
  exist some $t' \in \supp{\vec{x}}$ given by
  \begin{align*}
        \begin{pmatrix}
          p \\ q
        \end{pmatrix}, 
        \begin{pmatrix}
          p' \\ q'
        \end{pmatrix} \mapsto
        \begin{pmatrix}
          r \\ q
        \end{pmatrix}, 
        \begin{pmatrix}
          r' \\ q'
        \end{pmatrix}
  \end{align*}
  for some $q, q' \in Q_2$. It remains to show that $t' \in \preset{R}$ in order to
  prove that ~(\ref{eq:claim:trap_nonempty}) holds. Indeed, since $r \in P$ or $r' \in P$, 
  we have that $(r, q) \in R$ or $(r', q') \in R$, and thus $t' \in \preset{R}$.

  Now we show that ~(\ref{eq:claim:trap}) holds. To this end, let
  $t \in \postset{R} \cap \supp{\vec{x}} \subseteq
  \preset{P \times Q_2}$. There exist $p \in P$ and $q \in Q_2$ such
  that $(p, q) \in \preset{t}$. Moreover, $\vec{x}(t) > 0$. We must
  prove $t \in \preset{R}$. We consider two cases
  \begin{itemize}
    \item Assume $t \in S_2$. By definition of $S_2$, $t$ is of the
      form
      \begin{align*}
        \begin{pmatrix}
          p \\ q
        \end{pmatrix}, 
        \begin{pmatrix}
          p' \\ q'
        \end{pmatrix} \mapsto
        \begin{pmatrix}
          p \\ r
        \end{pmatrix}, 
        \begin{pmatrix}
          p' \\ r'
        \end{pmatrix}
      \end{align*}
      for some $p' \in Q_1$ and $q', r, r' \in Q_2$. In particular, we
      have $(p, r) \in \postset{t}$ which implies $t \in \preset{(P
        \times Q_2)} = \preset{R}$.

    \item Assume $t \in S_1$. Let $s \defeq \pi_1(t)$. By definition
      of $S_2$, $t$ is of the form
      \begin{align*}
        \begin{pmatrix}
          p \\ q
        \end{pmatrix}, 
        \begin{pmatrix}
          p' \\ q'
        \end{pmatrix} \mapsto
        \begin{pmatrix}
          r \\ q
        \end{pmatrix}, 
        \begin{pmatrix}
          r' \\ q'
        \end{pmatrix}
      \end{align*}
      where $\pre{s} = \multiset{p, p'}$, $\post{s} = \multiset{r,
        r'}$ and $q' \in Q_2$. This implies that $s \in \postset{p}
      \subseteq \postset{P}$. Moreover, since $t \in
      \supp{\vec{x}}$, we have $s \in
      \supp{\pi_1(\vec{x})}$. Therefore,
      by~(\ref{eq:trap:constraint}), we have $s \in \preset{P}$. This
      implies that either $r \in P$ or $r' \in P$, which in turn
      implies that $t \in \preset{R}$.
  \end{itemize}

  \noindent\emph{$U$-Siphon constraints:} Symmetric to $U$-trap constraints.
\end{proof}

\begin{proposition}\label{prop:conj:enabled}
  For every $i \in \{1, 2\}$, $C \in \pop{Q}$ and $t \in T_i$, $t$ is
  enabled in $\pi_i(C)$ if and only if there exists $s \in S_i$ such
  that $\pi_i(s) = t$ and $s$ is enabled in $C$.
\end{proposition}

\begin{proof}
  We only prove the claim for $i = 1$, as the case $i = 2$ is
  symmetric. Let $p, q \in Q_1$ be such that $\pre{t} = \multiset{p_1,
  q_1}$. By definition of $\pi_1$, we have
  \begin{align*}
    \pi_i(C)(p_1) &\defeq \sum_{p_2 \in Q_2} C(p_1, p_2), \text{ and} \\
    \pi_i(C)(q_1) &\defeq \sum_{q_2 \in Q_2} C(q_1, q_2).
  \end{align*}
  This implies that
  \begin{align}
    \pi_1(C) \geq \multiset{p_1, q_1} \iff \exists p_2, q_2 \in Q_2
    \text{ s.t. } C \geq \multiset{(p_1, p_2), (q_1,
      q_2)}.\label{eq:pre:proj}
  \end{align}

  \noindent$\Rightarrow$) Assume $t$ is enabled in $\pi_1(C)$. By
  (\ref{eq:pre:proj}), $C \geq \multiset{(p_1, p_2), (q_1, q_2)}$ for
  some $p_2, q_2 \in Q_2$. Let
  \begin{align*}
    s &\defeq t \ziptimes (p_2, q_2)
  \end{align*}
  We have $s \in S_1$. Moreover, $s$ is enabled at $C$.

  \noindent$\Leftarrow$) Assume there exists $s \in S_1$ such that
  $\pi_1(s) = t$ and $s$ is enabled at $C$. By definition of $S_1$,
  \begin{align*}
    s &= t \ziptimes (p_2, q_2)
  \end{align*}
  for some $p_2, q_2 \in Q_2$. Since $s$ is enabled at $C$, we have $C
  \geq \multiset{(p_1, p_2), (q_1, q_2)}$. By~(\ref{eq:pre:proj}),
  this implies $\pi_1(C) \geq \multiset{p_1, q_1}$, which in turn
  implies that $t$ is enabled at $\pi_1(C)$.
\end{proof}

\begin{corollary}\label{cor:term:proj}
  For every $C \in \pop{Q}$, if $C$ is terminal in $\PP$, then
  $\pi_1(C)$ and $\pi_2(C)$ are respectively terminal in $\PP_1$ and
  $\PP_2$.
\end{corollary}

\begin{proof}
  Let $C \in \pop{Q}$ be such that $C$ is terminal in $\PP$. For the
  sake of contradiction, suppose there exists $i \in \{1, 2\}$ such
  that $\pi_i(C)$ is not terminal in $\PP_i$. There exists $t \in T_i$
  such that $t$ is non silent and enabled in $\pi_i(C)$. By
  Proposition~\ref{prop:conj:enabled}, there exists $s \in S_i$ such
  that $\pi_i(s) = t$ and $s$ is enabled at $C$. We have $s = t \ziptimes q$
  for some $q \in Q_2 \times Q_2$. This implies that $s$ is non
  silent, since $t$ is non silent. We conclude that $C$ is non
  terminal which is contradiction.
\end{proof}

\begin{restatable}{lemma}{lemmaConjSplit}\label{lem:conj:split}
  If $\PP_1$ and $\PP_2$ satisfy \snosplitterminal, then $\PP$
  satisfies \snosplitterminal.
\end{restatable}

\begin{proof}
  We prove the contrapositive: if $\PP$ does not satisfy
  \snosplitterminal, then at least one of $\PP_1$ and $\PP_2$ does not
  satisfy \snosplitterminal. Assume $\PP$ does not satisfy
  \snosplitterminal. There are two cases to consider.
  \begin{itemize}
  \item[(a)] There exist $C, C' \in \pop{Q}$ such that $C \ptrans{}
    C'$, $C$ is initial, $C'$ is a terminal non consensus
    configuration. Since $C'$ is a non consensus configuration, there
    exist $(p, q), (p', q') \in \supp{C'}$ such that $O_1(p) \land
    O_1(q) = O(p, q) \not= O(p', q') = O_2(p') \land O_2(q')$. Without
    loss of generality, we can assume that $O_1(p) \not= O_1(p')$. By
    Corollary~\ref{cor:term:proj}, $\pi_1(C')$ is terminal in
    $\PP_1$. Moreover, since $p, p' \in \pi_1(C')$, $\pi_1(C')$ is a
    non consensus configuration. Therefore, $\pi_1(C')$ is a terminal
    non consensus configuration of $\PP_1$. Moreover, by
    Prop.~\ref{prop:conj:ptrans} $\pi_1(C) \ptrans{\pi_1(\vec{x})}
    \pi_1(C')$ which implies that $\PP_1$ does not satisfy
    \snosplitterminal.

  \item[(b)] There exist $C_0, C, C' \in \pop{Q}$ and $\vec{x},
    \vec{x}' : T \to \N$ such that $C_0 \ptrans{\vec{x}} C$, $C_0
    \ptrans{\vec{x}'} C'$, $C_0$ is initial, $C$ and $C'$ are terminal
    consensus configurations, and $O(C) \neq O(C')$. Since $C$ and
    $C'$ have different opinions, there exist $(p, q) \in \supp{C}$
    and $(p', q') \in \supp{C'}$ such that $O(p, q) \not= O(p',
    q')$. Without loss of generality, we can assume that $O_1(p) \not=
    O_1(p')$. By Corollary~(\ref{cor:term:proj}), $\pi_1(C)$ and
    $\pi_1(C')$ are terminal in $\PP_1$. Moreover, since $p
    \in\pi_1(C)$ and $p' \in \pi_1(C')$, $\pi_1(C)$ and $\pi_1(C')$
    are terminal configuration with different consensus. Moreover, by
    Prop.~\ref{prop:conj:ptrans}, $\pi_1(C) \ptrans{\pi_1(\vec{x})}
    \pi_1(C')$ which implies that $\PP_1$ does not satisfy
    \snosplitterminal. \qedhere
  \end{itemize}
\end{proof}

\begin{restatable}{proposition}{propConjLayered}\label{prop:conj:layered}
  If $\PP_1$ and $\PP_2$ satisfy \layeredtermination, then $\PP$
  satisfies \layeredtermination.
\end{restatable}

\begin{proof}
  Let $X_1, X_2, \ldots, X_m$ and $Y_1, Y_2, \ldots, Y_n$ be ordered
  partitions respectively for \layeredtermination in $\PP_1$ and
  $\PP_2$. We may assume without loss of generality that $m \geq
  n$. For every $n < i \leq m$, we define $Y_i \defeq \emptyset$.

  For every $i \in [m]$, we let
  \begin{alignat*}{2}
    Z_i\ & \defeq\ && \{t \ziptimes r : t \in X_i, r \in Q_2 \times Q_2\}
    \cup {} \\
    &&& \{r \ziptimes t : t
    \in Y_i, r \in Q_1 \times Q_1\}
  \end{alignat*}
  We claim that $Z_1, Z_2, \ldots, Z_m$ is an ordered partition for
  \layeredtermination in $\PP$. Let $i \in [m]$. Let us show that
  every execution of $\PP[Z_i]$ is silent. Suppose for the sake of
  contradiction that there exist $C_0, C_1, \ldots \in \pop{Q}$ and
  $t_0, t_1, \ldots \in Z_i$ such that $C_0 \trans{t_0} C_1
  \trans{t_1} \cdots$ is non silent. There exists $j \in \{1, 2\}$
  such that infinitely many non silent transitions $t_i$ belong to
  $S_j$. Let $i_0 < i_1 < \cdots$ be all indices such that $t_{i_k}
  \in S_j$. We have
  $$\pi_j(C_{i_0}) \trans{\pi_j(t_{i_0})} \pi_j(C_{i_1})
  \trans{\pi_j(t_{i_1})} \cdots$$ which is an infinite non silent
  execution of $\PP_1[X_i]$ or $\PP_2[Y_i]$ depending on $j$. This is
  a contradiction.

  Let $W \defeq (Z_1 \cup \cdots \cup Z_{i-1})$. Let us now prove that
  $\PP[Z_i]$ is $W$-dead. For the sake of contradiction, assume it is
  not. There exist $C, C' \in \pop{Q}$, $w \in Z_i^+$ and $t \in W$
  such that $C$ is $W$-dead, $C \trans{w} C'$ and $t$ is enabled at
  $C'$. We have $t \in S_j$ for some $j \in \{1, 2\}$. We may assume
  without loss of generality that $j = 1$. Since $C$ is $W$-dead,
  $\pi_j(C)$ is $(X_1 \cup \cdots \cup X_{i-1})$-dead. But then,
  $\pi_1(C) \trans{*} \pi_1(C')$ and $t \in X_1 \cup \cdots \cup
  X_{i-1}$ is enabled at $C'$ which is a contradiction.
\end{proof}

\begin{restatable}{corollary}{corConjWSSS}\label{cor:conj:wsss}
  If $\PP_1$ and $\PP_2$ belong to \WSSS, then $\PP$ belongs to \WSSS
  and is correct.
\end{restatable}

\begin{proof}
  By Lemma~\ref{lem:conj:split} and Prop.~\ref{prop:conj:layered},
  $\PP$ belongs to \WSSS. Let $w \in \pop{\Sigma}$, $C \defeq I(w)$,
  $C_1 \defeq I_1(w)$ and $C_2 \defeq I_2(w)$. Note that all three
  protocols are well-specified since they belong to \WSSS. Therefore,
  there exist terminal consensus configurations $C' \in \pop{Q}$,
  $C_1' \in \pop{Q_1}$ and $C_2' \in \pop{Q_2}$ such that $C \trans{*}
  C'$, $C_1 \trans{*} C_1'$ and $C_2 \trans{*} C_2'$.

  We must prove that $O(C') = O_1(C_1') \land O_2(C_2')$. Let $j \in
  \{1, 2\}$. Since $C \trans{*} C'$, we have $C \ptrans{} C'$. By
  Prop.~\ref{prop:conj:ptrans}, $\pi_j(C) \ptrans{} \pi_j(C')$. By
  definition of $I$, we have $C_j = \pi_j(C)$. Therefore, $C_j
  \ptrans{} \pi_j(C')$. Moreover, by Corollary~\ref{cor:term:proj},
  $\pi_j(C')$ is terminal in $\PP_j$. Since $\PP_j$ belongs to \WSSS,
  $\pi_j(C')$ is a consensus configuration such that $O_j(\pi_j(C')) =
  O_j(C_j')$. Altogether, we obtain
  \begin{align*}
    O(C') &= O_1(\pi_1(C')) \land O_2(\pi_2(C')) && \text{(by def. of $O$)} \\
          &= O_1(C_1') \land O_2(C_2'). &&\qedhere
  \end{align*}
\end{proof}

\section{Full set of constraints for Section~\ref{sec:experimental}}

We detail the constraints tested with the SMT solver in our
implementation.
Given a population protocol $\PP = (Q, T, \Sigma, I, O)$,
let $U \defeq \{ t \in T \colon \post{t} \neq \pre{t} \}$ be
the set of non silent transitions.

\subsection{Constraints for \layeredtermination}

For a given number of layers $k \in [1,|T|]$, the constraints for
\layeredtermination use variables $\vec{y}_i \colon Q \rightarrow \N$
for each $i \in [1,k]$ and $\vec{b} \colon T \rightarrow \N$. Vector
$\vec{b}$ assigns transitions to layers.

For a pair of transitions $t,u$, define
$U'(t, u) \defeq \{ u' \in U \colon \pre{u'} \le \pre{t} + (\pre{u} \mminus \post{t})\}$.
These sets are precomputed. The full set of constraints is:
\begin{align*}
    \bigwedge_{i \in [1,k]} \bigwedge_{t \in U} & \vec{b}(t) = i \rightarrow
        \sum_{q \in Q} \vec{y}_i(q) \cdot (\post{t}(q) - \pre{t}(q)) < 0 \\
    \bigwedge_{t \in T} & 1 \le \vec{b}(t) \le k \quad\land\quad
    \bigwedge_{t \in T} \bigwedge_{u \in U} \vec{b}(u) < \vec{b}(t) \rightarrow \bigvee_{u' \in U'(t,u)} \vec{b}(u) = \vec{b}(u')
\end{align*}

If there is a solution, then each layer $i \in [1,k]$ is given by
$T_i \defeq \{ t \in T \mid \vec{b}(t) = i \}$, and
each vector $\vec{y}_i$ assigns
a valuation $\vec{y}_i(C) \defeq \sum_{q \in Q} \vec{y}_i(q) \cdot C(q)$ to a given configuration $C$.
For any configuration $C$, we have $\vec{y}_i(C) \ge 0$ and
for any $C \trans{t} C'$ with $t \in T_i$, we have $\vec{y}_i(C) \ge \vec{y}_i(C')$ and
additionally $\vec{y}_i(C) > \vec{y}_i(C')$ if $t$ is non silent.
This is a proof that every execution of $\PP[T_i]$ is silent.

\subsection{Constraints for \snosplitterminal}

For a vector $\vec{c} : Q \rightarrow \N$, we define the following
constraints:
\begin{align*}
    \text{Initial}(\vec{c})  &\defeq \sum_{q \in I(\Sigma)} \vec{c}(q) \ge 2 \land \sum_{q \in Q \setminus I(\Sigma)} \vec{c}(q) = 0 &
    \text{True}(\vec{c})     &\defeq \sum_{q \in O^{-1}(1)} \vec{c}(q) > 0 \\
    \text{Terminal}(\vec{c}) &\defeq \bigwedge_{t \in U} \bigvee_{q \in \preset{t}} \vec{c}(q) < \pre{t}(q) &
    \text{False}(\vec{c})    &\defeq \sum_{q \in O^{-1}(0)} \vec{c}(q) > 0
\end{align*}

For sets of states $R,S \subseteq Q$ and vectors $\vec{c}, \vec{c'} :
Q \rightarrow \N$, and $\vec{x} : T \rightarrow \N$, we define the
following constraints:
\begin{align*}
    \text{FlowEquation}(\vec{c}, \vec{c'}, \vec{x})  &\defeq
        \bigwedge_{q \in Q} \vec{c'}(q) = \vec{c}(q) + \sum_{t \in T} \vec{x}(t) \cdot (\post{t}(q) - \pre{t}(q)) \\
    \text{UTrap}(R, \vec{c}, \vec{c'}, \vec{x})  &\defeq
        \sum_{t \in \preset{R}} \vec{x}(t) > 0 \land \sum_{t \in \postset{R} \setminus \preset{R}} \vec{x}(t) = 0 \rightarrow
        \sum_{q \in R} \vec{c'}(q) > 0 \\
    \text{USiphon}(S, \vec{c}, \vec{c'}, \vec{x})  &\defeq
        \sum_{t \in \postset{S}} \vec{x}(t) > 0 \land \sum_{t \in \preset{S} \setminus \postset{S}} \vec{x}(t) = 0 \rightarrow
        \sum_{q \in S} \vec{c}(q) > 0
\end{align*}

The constraints for \snosplitterminal use the variables $\vec{c}_0, \vec{c}_1,\vec{c}_2 \colon Q \rightarrow \N$
and $\vec{x}_1, \vec{x}_2 \colon T \rightarrow \N$.
With given sets $\mathcal{R}$ of $U$-traps and $\mathcal{S}$ of $U$-siphons, the constraints are as follows:

\begin{align*}
    & \text{FlowEquation}(\vec{c}_0, \vec{c}_1, \vec{x}_1) \quad\land\quad \text{FlowEquation}(\vec{c}_0, \vec{c}_2, \vec{x}_2) \\
    & \text{Initial}(\vec{c}_0) \quad\land\quad \text{Terminal}(\vec{c}_1) \quad\land\quad \text{Terminal}(\vec{c}_2) \quad\land\quad \text{True}(\vec{c}_1) \quad\land\quad \text{False}(\vec{c}_2) \\
    & \bigwedge_{R \in \mathcal{R}}
      \text{UTrap}(R, \vec{c}_0, \vec{c}_1, \vec{x}_1) \quad\land\quad \text{UTrap}(R, \vec{c}_0, \vec{c}_2, \vec{x}_2) \\
    & \bigwedge_{S \in \mathcal{S}}
      \text{USiphon}(S, \vec{c}_0, \vec{c}_1, \vec{x}_1) \quad\land\quad \text{USiphon}(S, \vec{c}_0, \vec{c}_2, \vec{x}_2)
\end{align*}

\end{document}